\documentclass[11pt,a4paper]{article}
\usepackage[margin=24mm]{geometry}
\usepackage[T1]{fontenc}
\usepackage{lmodern,microtype}
\usepackage{amsmath,amssymb,amsthm,mathtools}
\usepackage{enumitem,booktabs,array,xcolor,needspace,hyperref}
\usepackage{tikz}
\hypersetup{colorlinks=true,linkcolor=blue!55!black,citecolor=blue!55!black,urlcolor=blue!55!black}
\hypersetup{
  pdftitle={Step Recursion: Rational Strides and Exact Inclusion Orders over the Reals},
  pdfauthor={Kirill Osipov},
  pdfsubject={Rational stride classification and bounded-intersection reductions to integer predicate equality},
  pdfkeywords={step recursion, rational strides, bounded recursion, grid codes, exact intersections, Grzegorczyk predicate classes, arithmetical hierarchy}
}
\newtheorem{theorem}{Theorem}[section]
\newtheorem{lemma}[theorem]{Lemma}
\newtheorem{proposition}[theorem]{Proposition}
\newtheorem{corollary}[theorem]{Corollary}
\newtheorem{remark}[theorem]{Remark}
\newtheorem{example}[theorem]{Example}
\newtheorem{definition}[theorem]{Definition}
\newcommand{\E}{\mathsf E}
\newcommand{\N}{\mathbb N_0}
\newcommand{\C}{\mathsf C}
\newcommand{\nz}{\operatorname{nz}}
\newcommand{\zz}{\operatorname{z}}
\newcommand{\NN}{\mathbb N}
\newcommand{\RR}{\mathbb R_{\ge0}}
\newcommand{\B}{\mathcal B}
\newcommand{\U}{\mathcal U}

\newcommand{\rh}[1]{\rho_{#1}}
\newcommand{\Dt}[1]{D_{#1}}

\DeclareMathOperator{\den}{den}
\DeclareMathOperator{\lcmop}{lcm}
\DeclareMathOperator{\gcdop}{gcd}
\title{Step Recursion: Rational Strides and\\Exact Inclusion Orders over the Reals}
\author{Kirill Osipov}
\begin{document}\maketitle
\begin{abstract}
We study finite-terminating, address-only step recursion on the nonnegative reals, from zero, unit successor and projections. For reduced positive rational strides $s=a/q$ and $t=b/v$, we prove $\mathcal C_s\subseteq\mathcal C_t$ if and only if $q$ divides $v$, for both the bounded and unbounded algebras $\mathcal C\in\{\mathcal B,\mathcal U\}$. The bounded compiler preserves earlier global majorants on every auxiliary input. Mixed rational stride sets are classified by least common denominators and supernatural denominators. Primitive-recursive translated-coordinate grid codes characterize the unbounded algebras and yield $\mathcal U_{1/q}\cap\mathcal U_{1/v}=\mathcal U_{1/\gcd(q,v)}$. For bounded intersections, we prove an exact reduction: all corresponding identities hold if and only if the single identity $\mathcal B_{1/2}\cap\mathcal B_{1/3}=\mathcal B_1$ holds, equivalently if the Boolean parts of the integer successor and addition algebras under scalar bounded recursion coincide. A Boolean reconstruction reduces general functions to predicates; a clipped fine-grid simulation with a frozen global resource bound identifies the union of predicate precision classes with the addition level. These reductions do not decide the resulting predicate equality, so literal bounded intersections remain unresolved. We also prove that unbounded-program equality is $\Pi^0_1$-complete and extensional bounded-algebra membership is $\Sigma^0_2$-complete. The model has no free real arithmetic or Zeno limits; semantic majorant validity is distinguished from syntax transformation.
\end{abstract}

\section{Introduction}
Step recursion replaces the integer predecessor by a designated descent
map while allowing each transition to see only the predecessor
\emph{address}. Discrete versions distinguish schedules with similar
asymptotic depths \cite{Osipov2026}. Here the domain is the nonnegative
half-line and each descent is a truncated rational translation.

Our minimal basis consists only of zero, unit successor and projections.
Every program is locally a constant or a translated input coordinate
(Lemma~\ref{lem:affine}), and every bounded program has a global
one-coordinate bound (Lemma~\ref{lem:onecoord}); nevertheless, unbounded
recursion simulates all primitive-recursive integer programs
(Lemma~\ref{lem:pr-syntax}). This isolates the effect of the recursion
operator from that of free real arithmetic. The translation stride $s$ is a distance
on the half-line; the integer strides in the discrete companion papers
count iterates of a fixed step function. The denominator order proved
here concerns the former parameter. Definability of a predecessor
alone does not imply closure under recursion along it, and a bounded
simulation must supply earlier majorants on \emph{every} auxiliary input,
not merely on its intended simulation diagonal.

\paragraph{Antecedents and contribution.}
Sparse iteration and bounded recursion have a classical integer lineage:
the LOOP characterization of Meyer and Ritchie
\cite{MeyerRitchie1967}, restricted iteration and unary bases
\cite{GermanoMazzanti1988,Severin2008,Mazzanti2013,Mazzanti2005}, and
comparisons of primitive sets and restricted programs
\cite{Kfoury1980,IbarraLeininger1981,Matos2015}.
Our integer simulation adapts these methods to real grids; sparse
programming and affine invariants themselves are not novelty claims.

The real models based on differential recursion
\cite{Moore1996,Campagnolo2002,Costa2009}, effective representations,
approximation or lifting
\cite{BournezHainry2006,Brattka1996,CampagnoloOjakian2007,CampagnoloOjakian2008},
and real arithmetic \cite{BlumShubSmale1989} have different primitives.
Our exact grid descriptors are not approximations or real names.
Discontinuous tests are definable, so this is not a Type-2 model of total
computable real functions. Timed-automaton regions
\cite{AlurDill1994} provide a geometric analogy rather than the grid
characterization proved here.

Subrecursive equivalence, indexing and extensional recognition are also
established subjects
\cite{ConstableMuchnick1972,Kozen1980,Hoyrup2015,Matos2016}; our complexity
claims concern the precise index sets of this language. Supernatural
parametrization of rational subgroups is classical
\cite{ConnesConsani2016}. The successor and addition predicate levels
used below have the explicit scalar-recursion definitions of
Section~\ref{sec:bounded-meets}, matching Clote's exposition
\cite{Clote1999}.

The specific contribution is the exact rational denominator order in
both recursion regimes, including a bounded compiler preserving earlier
global majorants, the primitive-recursive grid characterization and
literal unbounded intersections, and the bounded-intersection equivalence
in Theorem~\ref{thm:equivalence}. Mixed normalization and supernatural
classification follow from the operator theorem and finite program
support.

\begin{theorem}[Main results; precise scope]\label{thm:main}
Let $s=a/q$, $t=b/v$ be positive rationals in lowest terms. Then, in the model of Definition~\ref{def:algebras},
\begin{equation}\label{eq:main}
\B_s\subseteq\B_t\iff q\mid v,\qquad
\U_s\subseteq\U_t\iff q\mid v.
\end{equation}
For every nonempty finite rational stride set $Q$, with $L(Q)=\lcmop\{\den(s):s\in Q\}$,
\begin{equation}\label{eq:mixmain}
\B_Q=\B_{1/L(Q)},\qquad \U_Q=\U_{1/L(Q)}.
\end{equation}
For any positive integers $q,v$, writing $d=\gcdop(q,v)$,
\begin{equation}\label{eq:meetmain}
\U_{1/q}\cap\U_{1/v}=\U_{1/d}.
\end{equation}
In addition, the following assertions, each over all positive finite arities,
are equivalent (Theorem~\ref{thm:equivalence}):
\begin{align*}
 &\forall q,v\geq1:\quad
   \B_{1/q}\cap\B_{1/v}=\B_{1/\gcdop(q,v)},\\
 &\B_{1/2}\cap\B_{1/3}=\B_1,\\
 &\E_{0,*}=\E_{1,*}.
\end{align*}
Here $\E_0=[0,S,\mathrm{proj};\mathrm{comp},\mathrm{BR}]$ and
$\E_1=[0,S,\mathrm{proj},+;\mathrm{comp},\mathrm{BR}]$ are integer
algebras with scalar bounded recursion, and the star denotes only their
Boolean-valued functions. The equivalence does not decide these assertions.

The single-stride classes in each regime form the positive-integer divisibility order. The bounded and unbounded regimes remain distinct for every nonempty rational stride set. These claims do not extend by assertion to nonlinear, irrational or nonterminating descents. For every rational stride, equality of unbounded programs is $\Pi^0_1$-complete, and the index set of unbounded programs whose functions belong to the bounded algebra is $\Sigma^0_2$-complete (Theorems~\ref{thm:equiv} and~\ref{thm:bmem}, and Corollary~\ref{cor:indexall}).
\end{theorem}

Figure~\ref{fig:rational-order} illustrates the denominator order and
distinguishes mixed closure from literal intersection.
\begin{figure}[htbp]
\centering
\begin{tikzpicture}[x=1cm,y=1cm,>=stealth,
  every node/.style={font=\small},
  algebra/.style={draw=black!65,rounded corners=2pt,fill=black!3,
    minimum width=15mm,minimum height=7mm,inner sep=3pt}]
  \node[font=\small\bfseries] at (0,4.25) {Denominator divisibility};
  \node[algebra] (c1) at (0,0) {$\mathcal C_1$};
  \node[algebra] (c2) at (-1.55,1.15) {$\mathcal C_{1/2}$};
  \node[algebra] (c3) at (1.55,1.15) {$\mathcal C_{1/3}$};
  \node[algebra] (c4) at (-1.55,2.30) {$\mathcal C_{1/4}$};
  \node[algebra] (c6) at (1.55,2.30) {$\mathcal C_{1/6}$};
  \node[algebra] (c12) at (0,3.45) {$\mathcal C_{1/12}$};
  \foreach \a/\b in {c1/c2,c1/c3,c2/c4,c2/c6,c3/c6,c4/c12,c6/c12}
    \draw[->,semithick] (\a) -- (\b);
  \node[align=center,font=\footnotesize] at (0,-.9)
    {$\mathcal C\in\{\mathcal B,\mathcal U\}$\\Arrows point to larger algebras.};
  \draw[black!25] (2.8,-.8) -- (2.8,4.45);
  \node[anchor=west,font=\small\bfseries] at (3.35,3.9)
    {Mixed closure: least common denominator};
  \node[anchor=west,align=left] at (3.35,3.15)
    {$\displaystyle\mathcal C_{\{1/2,1/3\}}=\mathcal C_{1/6}$\\[3pt]
     $\operatorname{lcm}(2,3)=6$};
  \node[anchor=west,font=\small\bfseries] at (3.35,2.05)
    {Literal unbounded intersection: gcd};
  \node[anchor=west,align=left] at (3.35,1.3)
    {$\displaystyle\mathcal U_{1/2}\cap\mathcal U_{1/3}=\mathcal U_1$\\[3pt]
     $\gcd(2,3)=1$};
  \node[anchor=west,align=left,font=\footnotesize] at (3.35,.05)
    {The bounded equality\\[2pt]
     $\mathcal B_{1/2}\cap\mathcal B_{1/3}\stackrel{?}{=}\mathcal B_1$\\[2pt]
     remains open.};
\end{tikzpicture}
\caption{The denominator order: arrows are strict inclusions and
transitive arrows are omitted. Mixed closure uses lcm in both regimes;
literal unbounded intersection uses gcd. The bounded case is governed by
Theorem~\ref{thm:equivalence}.}
\label{fig:rational-order}
\end{figure}

\section{Definition of the algebras}\label{sec:setup}
Put $\NN_0=\{0,1,2,\dots\}$ and $\RR=[0,\infty)$. For $s>0$ define
\[
 \rh{s}(y)=(y-s)_+,\qquad \Dt{s}(0)=0,\qquad
 \Dt{s}(y)=\lceil y/s\rceil\quad(y>0).
\]
Thus $\rh{s}^{[j]}(y)=(y-js)_+$, and the orbit of any finite $y$ reaches zero in exactly $D_s(y)$ steps. For $c\ge0$ write $T_c(x)=x+c$; this notation does not assert that $T_c$ belongs to a specified algebra.

\paragraph{Step functions, descent functions and depth.}
We use the terminology of the discrete step-recursion setting
\cite{Osipov2026}: the translation $\Phi_s(x)=x+s$ is the
\emph{step function}, and
\[
 \rho_s(y)=\min\{z\geq0:\Phi_s(z)\geq y\}
\]
is its generalized inverse, or \emph{descent function}. For any map $f$,
$f^{[0]}=\operatorname{id}$ and $f^{[j+1]}=f\circ f^{[j]}$ denote
finite iterates. The \emph{descent depth} is
$D_{\rho_s}(y)=\min\{j:\rho_s^{[j]}(y)=0\}$, abbreviated $D_s(y)$.
Thus the transition's predecessor address is $\rho_s(y)$.

\begin{definition}[Address-only finite-step recursion]\label{def:algebras}
From zero, $S(x)=x+1$, and all positive finite-arity projections, form finite-stage closures under composition and the following operation along each allowed stride $s$:
\begin{equation}\label{eq:rec}
F(\mathbf x,0)=a(\mathbf x),\qquad
F(\mathbf x,y)=b(\mathbf x,\rh{s}(y),F(\mathbf x,\rh{s}(y)))\quad(y>0).
\end{equation}
\emph{Unbounded step recursion} admits the unique total solution of these terminating equations; its closure for a set $Q$ of strides is $\U_Q$. \emph{Bounded step recursion} admits the same solution only if a function $M$ constructed at an earlier stage satisfies $0\le F(\mathbf x,y)\le M(\mathbf x,y)$ on all inputs. Its closure is $\B_Q$. We abbreviate a singleton $Q=\{s\}$ as $\B_s,\U_s$.
We also write $\mathrm{USR}_{\rho_s}(a,b)=F$ for the unbounded
operator and $\mathrm{SR}_{\rho_s}(a,b,M)=F$ for the bounded operator.
Here bounded always means \emph{previously majorized}: the global
pointwise majorant belongs to an earlier construction stage. Accordingly,
$\mathcal B_Q$ is the previously majorized algebra, also denoted
$\mathcal B_Q^{\mathrm{prev}}$ when that convention needs emphasis.
The initial signature will be called the \emph{minimal basis}.
\end{definition}

All maps have fixed positive finite arity. Dummy parameters are allowed. \emph{Unbounded} refers solely to the absence of a previous-majorant requirement: all descents in this paper are finite. Semantic admission of a majorant is not claimed to be decidable. Fixed real stride values parameterize operators but are \emph{not} free real constants in the initial function basis.

A \emph{conditional} between previously defined expressions $A$ and $B$ controlled by a nonnegative $c$ is implemented by a step recursion on $c$ whose base is $A$ and whose transition is $B$ independent of the recursive value. It is therefore available without restriction in $\U_Q$. In $\B_Q$ it is permitted only if an earlier majorant bounds the \emph{selected} value globally. Where needed below we justify that majorant explicitly.

\subsection*{Finite syntax and effective interpretation}
The finite-stage convention means that each expression has a finite syntax tree with children created at earlier stages. In a bounded-recursion node $(a,b,M)$, the majorant $M$ is part of the witness: its inequality must hold for all real arguments, including auxiliary arguments introduced by a compiler. We use stage-ordered \emph{semantic} bounded descriptions rather than assuming decidable syntactic checking of their majorant conditions. For rational strides, evaluation of any unbounded expression on rational inputs is exact and terminates; this fact will be important when defining index sets.

\section{The affine invariant and predecessor obstruction}\label{sec:invariant}
For $Q\subset(0,\infty)$ finite let
\[
 M_Q=\left\{m-\sum_{s\in Q}n_s s:\ m,n_s\in\NN_0\right\}.
\]
It is an additive monoid, not generally a group.

\begin{lemma}[Local translated-coordinate invariant]\label{lem:affine}
On each compact box every $f\in\U_Q$ admits a finite partition into axis-aligned cells on which it equals either a constant $c$ or $x_i+c$, with $c\in M_Q$. This holds a fortiori for $\B_Q$.
\end{lemma}
\begin{proof}
The initial functions have the indicated forms. For composition, subdivide a compact input box according to the finitely many translated-coordinate descriptions of the inner functions, then refine by inverse images of the axis-aligned cells describing the outer function. Each inverse image under a constant or $x_i+c$ is again a coordinate interval condition. Substituting the inner affine descriptions adds offsets in $M_Q$.

For a recursion along $\rho_s$, the unfolding depth is uniformly bounded on a compact box in the recursion variable. Each finite unfolding uses composition and the distinction $y=0$ versus $y>0$ (an axis-aligned cut). The finitely many unfoldings therefore retain the stated property, and the actual recursion agrees with one such finite unfolding on the fixed box. Induction over program stages completes the proof.
\end{proof}

\begin{lemma}[Exact predecessor test]\label{lem:pred}
For $s>0$ and finite nonempty $Q$,
\begin{equation}\label{eq:predtest}
 \rh{s}\in\U_Q\iff\rh{s}\in\B_Q
 \iff s=\sum_{t\in Q}n_t t-m>0
 \quad(n_t,m\in\NN_0).
\end{equation}
In particular, for $Q=\{t\}$ the last condition becomes $s=nt-m>0$ for some $n\ge1,m\ge0$.
\end{lemma}
\begin{proof}
If $\rho_s$ is definable, its affine portion $x\mapsto x-s$ on an open interval must coincide with a translated projection allowed by Lemma~\ref{lem:affine}; hence $-s\in M_Q$. Conversely, $\rho_t$ has a bounded one-step recursion definition (base zero, transition the predecessor address, earlier majorant the recursion argument). Repeated composition builds $\rho_{\sum n_t t}$, and composing it with $S^{[m]}$ gives $\rho_s$ whenever $s=\sum n_t t-m>0$.
\end{proof}

\begin{corollary}[Necessary rational denominator condition]\label{cor:necessary}
If $s=a/q$ and $t=b/v$ are reduced positive rationals and $\U_s\subseteq\U_t$ or $\B_s\subseteq\B_t$, then $q\mid v$.
\end{corollary}
\begin{proof}
The inclusion makes $\rho_s$ definable from $t$. By Lemma~\ref{lem:pred}, $a/q=nt-m$. Thus $av=q(nb-mv)$ and $q\mid v$, since $\gcd(a,q)=1$.
\end{proof}

\begin{lemma}[One-coordinate global bound]\label{lem:onecoord}
For every $f:\RR^k\to\RR$ in $\B_Q$ there are an input coordinate $i$ and an integer $C\ge0$ such that
\[f(\mathbf x)\le x_i+C\qquad\text{for all }\mathbf x.\]
\end{lemma}
\begin{proof}
The argument also extracts a pair $(i,C)$ effectively from a finite
stage-ordered candidate bounded description. For zero choose $(1,0)$,
for successor choose $(1,1)$, and for the $i$th projection choose
$(i,0)$. If an outer function $g$ has extracted bound $v_j+C_g$
and its $j$th inner function has extracted bound $x_i+C_j$, assign
the composite the pair $(i,C_g+C_j)$. At a bounded-recursion node,
copy the pair extracted from its earlier majorant. Structural induction
proves the resulting bound whenever all annotations are semantically
valid. Extraction is a syntactic computation even on invalid candidates;
it does not decide validity of their majorant annotations.
\end{proof}

\section{Rational tools and the two bounded compilers}\label{sec:tool}
\begin{lemma}[Rational translations and finite controls]\label{lem:toolkit}
If $t=a/q>0$ is reduced, $\B_t$ contains $\rho_{k/q}$, $T_{k/q}$ and the constant $k/q$ for all positive integers $k$. It also contains tests of fixed $1/q$-grid thresholds and finite-valued residue clocks $D_t(y)\bmod n$ for fixed $n$.
\end{lemma}
\begin{proof}
The predecessor $\rho_t$ is obtained by bounded recursion with bound $y$. The equations $na-mq=\pm k$ have solutions $n,m\in\NN_0$ with $n\ge1$ when necessary, by Bezout's identity followed by adding sufficiently large multiples of $(q,a)$ to make both variables nonnegative. Therefore $\rho_t^{[n]}(x+m)$ is either $\rho_{k/q}(x)$ or $T_{k/q}(x)$, with the signs chosen accordingly. Evaluating positive translations at zero gives constants. The nonzero test is a bounded recursion with constant bound one; composing it with predecessors yields fixed thresholds. A bounded finite-state recursion along $\rho_t$ cycles through $n$ constant states and yields $D_t(y)\bmod n$.
\end{proof}

\begin{lemma}[Bounded finite selection]\label{lem:select}
If previously defined $f_0,f_1,c,H$ satisfy $0\le f_0,f_1\le H$ globally, the branch selecting $f_0$ when $c=0$ and $f_1$ when $c>0$ belongs to the same bounded algebra. A branch $f$ which is bounded by $H$ only on an earlier definable predicate's active region may first be zero-guarded using a previously available bound for $f$.
\end{lemma}
\begin{proof}
Recursion on $c$ with base $f_0$ and constant transition $f_1$ has previous majorant $H$. For a zero-guard, select zero and $f$ with an existing global bound for $f$; the resulting function is bounded by $H$ because it vanishes off the active region. Thus later recombinations may use $H$ as an earlier bound.
\end{proof}

\begin{lemma}[Frozen-parameter guards and stage separation]\label{lem:frozen}
Suppose previously constructed expressions $u_0(\mathbf v),u_1(\mathbf v),c(\mathbf v),H(\mathbf v)$ satisfy $0\le u_0,u_1\le H$ on the full input domain. Then the selector $u_0$ if $c=0$, and $u_1$ otherwise, is constructed without a new comparison primitive by a bounded recursion on an \emph{additional} variable $z$, with all of $\mathbf v$ frozen, and evaluation at $z=c(\mathbf v)$. The earlier majorant for the entire auxiliary recursion is $H(\mathbf v)$, independently of $z$. In particular, if a branch $u_1$ is required only on the region $c>0$ and there satisfies $u_1\le H$, first form its globally zero-guarded version using its own earlier global majorant; this guarded function is everywhere bounded by $H$ and may then be combined with further branches under $H$.
\end{lemma}
\begin{proof}
Put $V(\mathbf v,0)=u_0(\mathbf v)$ and $V(\mathbf v,z)=u_1(\mathbf v)$ for $z>0$. The transition ignores the recursive state and predecessor address. Both clauses are bounded by $H(\mathbf v)$ on \emph{every} $(\mathbf v,z)$, so the bounded operator is admissible. Evaluation at $c(\mathbf v)$ gives the selector. Apply the same construction with the zero branch and the earlier bound for $u_1$ to obtain a zero-guarded branch. The claimed bound by $H$ is then a separate pointwise property of the already constructed guarded function, not an assumption used to admit its definition.
\end{proof}

\begin{lemma}[Globally bounded finite unfoldings]\label{lem:unfold}
Let $u>0$ be rational. Suppose an ambient algebra contains zero,
projections, $\rho_u$ and the fixed $u$-grid band predicates, and is
closed under composition and the bounded selectors of
Lemma~\ref{lem:frozen}. Let $a(\mathbf x)$, $b(\mathbf x,z,v)$
and $M(\mathbf x,r)$ be previously constructed functions in that
algebra. Suppose the unique terminating $u$-recursion solution $F$
with these data satisfies $F\le M$ on all inputs; closure under
bounded $u$-recursion is not assumed. For a fixed integer $p\geq1$
put $t=pu$. There is an earlier function $B(\mathbf x,r)$ in the
ambient algebra such that
\[
 B(\mathbf x,r)=F(\mathbf x,r)\ (0\leq r\leq t),\qquad
 B(\mathbf x,r)=0\ (r>t),\qquad 0\leq B\leq M\ \text{everywhere}.
\]
No use of the as-yet-unconstructed function $F$ is made in constructing $B$.
\end{lemma}
\begin{proof}
For $0\leq j\leq p$, define a finite term $L_j(\mathbf x,r)$:
$L_0=a$, and for $j\geq1$ start at $v_0=a$ and successively put
\[
 v_{\ell+1} = b\bigl(\mathbf x,\rho_u^{[j-\ell]}(r),v_{\ell}\bigr)
 \quad(0\leq\ell<j),\qquad L_j=v_j.
\]
This is a fixed finite composition of earlier expressions.  On the band
$(j-1)u<r\leq ju$ its successive predecessor addresses are $0$,
$\rho_u^{[j-1]}(r),\ldots,\rho_u(r)$; hence $L_j=F$ there.
For $j=0$, $L_0=F$ on $\{0\}$.  Let $P_j(r)$ be the band predicate,
constructed by finite predecessor tests, including $P_0=[r=0]$.
Since $L_j$ is already a previously constructed nonnegative function,
we may take $H_j=L_j$ as an earlier global majorant for its
zero-guard. This does not rely on the one-coordinate bound invariant
and applies in the general ambient algebra of Theorem~\ref{thm:block}.
First use the
bounded selector of Lemma~\ref{lem:frozen} to define
$G_j=L_j$ when $P_j=1$ and $G_j=0$ otherwise, with $H_j$ as the
admission bound of this \emph{first} selector.  The resulting $G_j$
has already been constructed and satisfies $0\leq G_j\leq M$ on the
entire real domain: on its active band this is the stipulated source
bound, and elsewhere the value is zero.
Finally combine $G_0,\ldots,G_p$ successively by bounded selection;
$M$, now an earlier bound for both branches at each stage, admits all
these selections.  Since exactly one band is active on $[0,t]$, the
result is $F$ there and zero elsewhere.  In particular the proof
never presupposes a bound on an \emph{unselected} $L_j$ by $M$.
\end{proof}

\begin{lemma}[Finite-state clock and stopped-address identities]\label{lem:clockdetail}
Fix rational $h>0$ and integer $n\geq1$, and put $s=nh$.  The
auxiliaries $E,R,r,Q$ of Theorem~\ref{thm:coarse} are constructed in
that order in $\B_h$, and satisfy, on the entire indicated domains,
\[
\begin{aligned}
 E(w)&=h\max(D_s(w)-1,0),&0\leq E(w)&\leq w+h,\\
 R(w,y)&=\rho_s^{[D_h(y)]}(w),&0\leq R(w,y)&\leq w,\\
 r(w)&=R(w,E(w)),& 0\leq r(w)&\leq w,\\
 Q(w,y)&=r(w)+s\min(D_h(y),\max(D_s(w)-1,0)),&0\leq Q(w,y)&\leq w.
\end{aligned}
\]
When $w=0$ the four states vanish; for $w>0$, $r(w)\in(0,s]$.
\end{lemma}
\begin{proof}
The phase counter $D_h(z)\bmod n$ is a bounded finite-state
recursion.  At a positive step with predecessor $z$, the clock $E$
increases by $h$ if $z>0$ and $D_h(z)\equiv0\pmod n$ and holds
otherwise.  If $m=D_h(w)$, the number of selected predecessors
$z$ of depths in $\{1,\ldots,m-1\}$ divisible by $n$ is
$\lfloor(m-1)/n\rfloor$, which equals $D_s(w)-1$ for $m>0$.
The earlier translation $w+\lceil h\rceil$ majorizes every value
of $E$; if $w=0$ the state remains zero.

The claimed formula for $R$ follows by induction on $D_h(y)$ from
its defining recursion, with global earlier bound $w$.
For positive $w$, $r(w)=w-s(D_s(w)-1)\in(0,s]$.
At the $j$th step of the recursion defining $Q$, the gate is
$[\rho_s(R(w,z))>0]$ with predecessor depth $D_h(z)=j-1$.
This gate is true exactly when $j\leq D_s(w)-1$.  Thus exactly
$\min(j,D_s(w)-1)$ increments of size $s$ have occurred.
Each transition value is bounded by $u+s$ for arbitrary state $u$;
this licenses the transition selector.  Since the resulting closed
form never exceeds $w$, the \emph{earlier projection} $w$ licenses
the entire recursion for $Q$.  These are distinct admission checks.
\end{proof}

\begin{example}[A finite block unfolding with global guards]\label{ex:unfold}
For $u=1/2$, $t=3/2$, the source recursion on $0\le r\le t$
is unfolded into the four branches $L_0=a$ at $r=0$,
$L_1=b(\mathbf x,0,a)$ on $(0,1/2]$,
$L_2=b(\mathbf x,\rho_u(r),b(\mathbf x,0,a))$ on
$(1/2,1]$, and
$L_3=b(\mathbf x,\rho_u(r),b(\mathbf x,\rho_u^{[2]}(r),b(\mathbf x,0,a)))$
on $(1,3/2]$.  Each expression is first zero-guarded under its own
previously known global bound; the resulting guarded functions all
obey the source majorant $M$, allowing their later recombination.
\end{example}

\begin{theorem}[Bounded coarse-to-fine simulation]\label{thm:coarse}
For $h\in\mathbb Q_{>0}$ and $n\ge1$,
\begin{equation}\label{eq:coarse}\B_{nh}\subseteq\B_h.\end{equation}
\end{theorem}
\begin{proof}
Write $s=nh$. Induct on the stages of a source program. Consider a source $s$-recursion $F$ with data $a,b,M$ already compiled into $\B_h$, and $F(\mathbf x,w)\le M(\mathbf x,w)$ for all inputs. We give a global-majorant-preserving compiler.

Put $K(w)=D_s(w)$, used only in mathematical descriptions, not as an assumed bounded-algebra function. The stage-order and global bounds of the clock, residual and stopped address are supplied by Lemma~\ref{lem:clockdetail}. Construct the compressed clock
\[
 E(w)=h\max(K(w)-1,0).
\]
A bounded $h$-recursion can compute $E$: when scanning $w$ from zero through its $h$-descent, increase a state by $h$ precisely when the predecessor address is positive and its $h$-depth is congruent to zero modulo $n$. The available finite residue clock tests this condition. If $m=D_h(w)>0$, the number of increments is $\lfloor(m-1)/n\rfloor=K(w)-1$. Its value is bounded by $w+h$, so an earlier integer translation of $w$ is a majorant.

With frozen parameter $w$ define $R(w,0)=w$ and
\[
 R(w,y)=\rho_s(R(w,\rho_h(y)))\quad(y>0).
\]
This $h$-recursion is bounded by $w$, and $R(w,y)=\rho_s^{[D_h(y)]}(w)$. Set $r(w)=R(w,E(w))$. For $w>0$,
\[r(w)=w-s(K(w)-1)\in(0,s],\qquad r(0)=0.\]
Define the gate $g(w,z)=[\rho_s(R(w,z))>0]$, and define an auxiliary address by
\[
 Q(w,0)=r(w),\quad Q(w,y)=
 \begin{cases}Q(w,z)+s,&g(w,z)>0,\\Q(w,z),&g(w,z)=0,\end{cases}
 \quad z=\rho_h(y).
\]
Both candidate updates have the earlier bound $u+s$ when the recursive state is $u$, so Lemma~\ref{lem:select} builds the transition. The resulting closed form is
\begin{equation}\label{eq:Qclock}
 Q(w,y)=\begin{cases}
 r(w)+s\min\bigl(D_h(y),K(w)-1\bigr),&w>0,\\
 0,&w=0.
 \end{cases}
\end{equation}
Hence $0\le Q(w,y)\le w$ \emph{for all} $w,y$; the earlier projection $w$ admits this bounded recursion. In particular $Q(w,E(w))=w$.

The positive-base value $b_*(\mathbf x)=b(\mathbf x,0,a(\mathbf x))$ is an earlier function. Let $B(\mathbf x,w)$ be $b_*(\mathbf x)$ for $w>0$ and zero for $w=0$, constructed by bounded zero-guarding. At positive $w$, $B=F(\mathbf x,r(w))$ since $0<r(w)\le s$. Lemma~\ref{lem:onecoord}, applied to the entire transition $b(\mathbf x,z,u)$, gives one of three global inequalities with integer $C$:
\begin{equation}\label{eq:trichotomy}
 b(\mathbf x,z,u)\le x_i+C,\qquad
 b(\mathbf x,z,u)\le z+C,\qquad\text{or}\qquad
 b(\mathbf x,z,u)\le u+C.
\end{equation}
The structural extraction in Lemma~\ref{lem:onecoord}, applied to the
compiled bounded description of $b$, supplies a coordinate and $C$,
and thus selects one of these cases effectively. No search for a true
universal inequality is required. For the first two define
\[
 J(\mathbf x,w,0)=B(\mathbf x,w),\quad
 J(\mathbf x,w,y)=b\bigl(\mathbf x,Q(w,z),J(\mathbf x,w,z)\bigr),\quad z=\rho_h(y).
\]
The earlier majorant is respectively $x_i+C$ or $w+C$: the base satisfies it and every update does too, using $Q(w,z)\le w$. Thus these are valid bounded $h$-recursions on the \emph{whole} auxiliary domain. When the third inequality holds, instead use the transition
\[
 J(\mathbf x,w,y)=\begin{cases}
 b\bigl(\mathbf x,Q(w,z),J(\mathbf x,w,z)\bigr),&g(w,z)>0,\\
 J(\mathbf x,w,z),&g(w,z)=0,
 \end{cases}
\]
with the same base. For any arbitrary candidate state $u$ both branches are bounded by $u+C$, so the transition is already in $\B_h$. In this third (state-dominated) case the gate freezes both $Q$ and $J$ after the last required update; induction gives $J(\mathbf x,w,y)=F(\mathbf x,Q(w,y))$ for every $w>0$ and every $y$, while $J(\mathbf x,0,y)=0$. Thus the earlier function $M(\mathbf x,Q(w,y))$ is a global majorant for $J$ in the third case (and at $w=0$ it dominates zero). In the first two cases the global simulation identity need not hold after the virtual address has stopped: it holds only on the required prefix $D_h(y)\le K(w)-1$, which includes $y=E(w)$. Their independent earlier majorants $x_i+C$ and $w+C$ remain valid on the entire auxiliary domain. Example~\ref{ex:prefix} gives a counterexample to the stronger identity after that prefix.

In every case the first $K(w)-1$ positive auxiliary stages follow the original $s$-recursion, from $r(w)$ to $w$, so $J(\mathbf x,w,E(w))=F(\mathbf x,w)$ for $w>0$. At $w=0$, return $a(\mathbf x)$ instead, using zero-guarded selection. First construct each branch with its own already known global majorant; then the selected function is pointwise at most the source majorant $M(\mathbf x,w)$ on every input, making $M$ a valid earlier bound for the final selector. The needed frozen-parameter stage separation is Lemma~\ref{lem:frozen}. This completes the source-stage induction.
\end{proof}

\begin{example}[Why the coarse compiler needs only a prefix invariant]\label{ex:prefix}
Let $s=2$, $h=1$, $a=0$, and $b(z,u)=z+1$, with source majorant $M(y)=y+1$. For $w=5/2$, the coarse recursion has $F(w)=3/2$. The compiler's parameters are $K(w)=2$, $r(w)=1/2$, $E(w)=1$, and $Q(w,0)=1/2$, $Q(w,1)=5/2$. It produces $J(w,0)=1$, $J(w,1)=3/2$, hence the correct final output. At an auxiliary input of depth $2$, however, the unfrozen address-dominated transition gives $J(w,y)=7/2$ while $F(Q(w,y))=3/2$. Thus the auxiliary identity is \emph{false} after the designated prefix. The genuinely global statement needed for admissibility is $J(w,y)\le w+1=7/2$; the simulation identity is asserted only on the initial prefix. This example prevents an unjustified strengthening of the induction hypothesis.
\end{example}

\begin{theorem}[Bounded fine-to-coarse block simulation]\label{thm:block}
Let $t=pu$, $p\in\NN_{>0}$, and let an algebra $\mathcal C$ contain $\rho_u$, $T_{ju}$ for $0\le j\le p$, the fixed $u$-threshold tests and $T_t$, and be closed under bounded $t$-recursion and composition. Then it is closed under the entire bounded $u$-recursion operator.
\end{theorem}
\begin{proof}
Take source data $a,b,M\in\mathcal C$ for a bounded $u$-recursion $F$. First define the positive $t$-residue $r_t(w)$, equal to zero at $w=0$ and to $w-t(D_t(w)-1)\in(0,t]$ for $w>0$. It is boundedly definable: run a $t$-recursion with frozen $w$, base $w+t$ and transition $v\mapsto\rho_t(v)$, bound it by $w+t$, evaluate at recursion argument $w$, and use a bounded zero-guard to set the value at $w=0$ to zero.

On $0\le r\le t=pu$, the result $F(\mathbf x,r)$ has at most $p$ source steps, so it is represented there by finitely many nested terms in $a,b$ and $\rho_u$. The global construction of the finite unfolding $B$ is provided by Lemma~\ref{lem:unfold}: each branch is guarded under its own earlier majorant before the source majorant $M$ is used for recombination. Thus $B$ already exists and agrees with $F$ on $[0,t]$.

Define the virtual address via a $t$-recursion with frozen $w$:
\[
 A(w,0)=r_t(w),\qquad A(w,y)=A(w,\rho_t(y))+t\quad(y>0).
\]
Then $A(w,y)=r_t(w)+tD_t(y)\le y+2t$ globally, so an earlier function $T_t(T_t(y))=y+2t$ bounds $A$. Put
\[
 C_p(\mathbf x,L,v)=b\bigl(\mathbf x,L+(p-1)u,
 b(\mathbf x,L+(p-2)u,\ldots b(\mathbf x,L,v)\ldots)\bigr).
\]
The macro uses only earlier functions. Define the $t$-recursion
\[
 J(\mathbf x,w,0)=B(\mathbf x,r_t(w)),\quad
 J(\mathbf x,w,y)=C_p\bigl(\mathbf x,A(w,z),J(\mathbf x,w,z)\bigr),\quad z=\rho_t(y).
\]
Induction over $D_t(y)$ gives $J(\mathbf x,w,y)=F(\mathbf x,A(w,y))$ on every auxiliary input. Consequently $M(\mathbf x,A(w,y))$ is a \emph{globally valid earlier majorant}, not merely a diagonal one. As $A(w,\rho_t(w))=w$ also at $w=0$, evaluation gives $J(\mathbf x,w,\rho_t(w))=F(\mathbf x,w)$ as required.
\end{proof}

\Needspace{16\baselineskip}
\subsection*{Compiler stage and majorant ledger}\label{sec:ledger}
For reference, the majorant-bearing auxiliaries used in Theorems~\ref{thm:coarse} and~\ref{thm:block} have the following creation order. All entries refer to \emph{already available} data when the next recursion is introduced.
\begin{center}\small
\begin{tabular}{@{}p{.26\linewidth}p{.24\linewidth}p{.37\linewidth}@{}}\toprule
Object & Earlier bound & Scope of the bound\\\midrule
$E(w)$ & $w+\lceil h\rceil$ & All $w\ge0$ \\
$R(w,y)$ & $w$ & All $(w,y)$ \\
$Q(w,y)$ & $w$ & All $(w,y)$, including after address saturation \\
$J$ (parameter-dominated) & $x_i+C$ & All auxiliary inputs \\
$J$ (address-dominated) & $w+C$ & All auxiliary inputs \\
$J$ (state-dominated) & $M(\mathbf x,Q(w,y))$ & All auxiliary inputs, using the gate which freezes the state \\
$r_t(w)$ & $w+t$ & All $w$, with a separate zero guard \\
$A(w,y)$ & $y+2t$ & All auxiliary inputs \\
$J$ (block compiler) & $M(\mathbf x,A(w,y))$ & All auxiliary inputs \\
\bottomrule\end{tabular}
\end{center}
Each bound is represented before the corresponding recursion is
introduced. The compiler is a computable transformation of candidate
descriptions; its correctness and the validity of the output annotations
are conditional on validity of the input annotations. The table records
global auxiliary bounds, rather than bounds only on the output diagonal.

\section{Exact rational inclusion order}\label{sec:class}
\begin{lemma}[Phase-accurate unbounded compiler]\label{lem:unbounded-phase}
For every $h>0$ and integer $n\ge1$, closure under unrestricted
$h$-recursion and composition, with the available finite phase tests,
implies closure under unrestricted $nh$-recursion.
\end{lemma}
\begin{proof}
Put $s=nh$ and take source data $a(\mathbf x)$ and
$b(\mathbf x,z,u)$. Freeze the source argument $w$ as an additional
parameter, and write $d=D_h(w)$. The number $d$ is used only in the
metatheoretic explanation; the program needs only its residue modulo
$n$. Define $G(\mathbf x,w,0)=a(\mathbf x)$. At a positive
microstep with predecessor $z=\rho_h(y)$, retain the previous state
unless
\[
 D_h(z)+1\equiv D_h(w)\pmod n.
\]
On an active step replace the state $u$ by
\(b(\mathbf x,\rho_h^{[n-1]}(z),u)\).
The gate, predecessor iterate and unbounded conditional are earlier
functions, so this is a valid unrestricted $h$-recursion.

For any $y\ge0$, set $m=D_h(y)$ and let
$r_m\in\{0,\ldots,n-1\}$ be the residue of $m-d$ modulo $n$.
Write $H_w(y)=\rho_h^{[r_m]}(y)$, again only as an invariant.
Induct over $m$ to show
\[
 G(\mathbf x,w,y)=F(\mathbf x,H_w(y)).
\]
The zero case holds because $H_w(0)=0$. If $m\not\equiv d\pmod n$,
then $r_m>0$ and
$H_w(y)=\rho_h^{[r_m-1]}(z)=H_w(z)$; the source value
is unchanged and the microstep is inactive. If $m\equiv d\pmod n$,
then $r_m=0$, $r_{m-1}=n-1$, and
$H_w(z)=\rho_h^{[n-1]}(z)=\rho_s(y)$.
The active update is exactly the source recursion equation at $y$.
Finally $H_w(w)=w$ since $m=d$, so the diagonal evaluation
$G(\mathbf x,w,w)$ equals $F(\mathbf x,w)$ for all inputs, including
$w=0$. No bound claim is needed.
\end{proof}

\begin{theorem}[Single-stride classification]\label{thm:classification}
If $s=a/q$ and $t=b/v$ are positive rationals in lowest terms, then
\[
\boxed{\B_s\subseteq\B_t\iff q\mid v,\qquad
\U_s\subseteq\U_t\iff q\mid v.}
\]
In either family equality holds exactly when $q=v$, and the two algebras are incomparable when neither denominator divides the other.
\end{theorem}
\begin{proof}
Necessity follows from Corollary~\ref{cor:necessary}. For the bounded sufficiency, Theorems~\ref{thm:coarse} and~\ref{thm:block}, using Lemma~\ref{lem:toolkit}, give $\B_{a/q}=\B_{1/q}$. If $q\mid v$, then $1/q=(v/q)(1/v)$, and Theorem~\ref{thm:coarse} gives $\B_{1/q}\subseteq\B_{1/v}=\B_t$.

For the unbounded case, Lemma~\ref{lem:unbounded-phase} supplies the whole operator simulation $\U_{nh}\subseteq\U_h$. Combining this simulation and the unrestricted version of Theorem~\ref{thm:block} gives $\U_{a/q}=\U_{1/q}$, and the same denominator-divisibility implication follows.
\end{proof}

\begin{corollary}[Distinct recursion regimes]\label{cor:strict}
For every nonempty stride set $Q\subseteq\mathbb Q_{>0}$, finite or infinite, $\B_Q\subsetneq\U_Q$.
\end{corollary}
\begin{proof}
Fix $s\in Q$. The unbounded recursion with base $a(x)=x$ and transition $b(x,z,u)=S(u)$ defines $F(x,y)=x+D_s(y)$. This is not bounded by $x+C$ (vary $y$) or $y+C$ (vary $x$), hence is excluded from $\B_Q$ by Lemma~\ref{lem:onecoord}.
\end{proof}

\section{Mixed strides and supernatural denominators}\label{sec:mixed}
\begin{lemma}[Synthesis of a common rational unit]\label{lem:bezout}
Let $Q=\{a_1/q_1,\ldots,a_k/q_k\}$ consist of positive reduced rationals. Put $L=\lcmop(q_1,\ldots,q_k)$ and $u=1/L$. Then $\rho_u$ and $T_u$ belong to $\B_Q$.
\end{lemma}
\begin{proof}
By Lemma~\ref{lem:toolkit}, $\rho_{1/q_i}$ and $T_{1/q_i}$ are in $\B_Q$. As $\gcd(L/q_1,\ldots,L/q_k)=1$, there are integers $c_i$ such that $\sum_i c_i/q_i=1/L$. Separating positive and negative coefficients gives nonnegative rationals $P,N$, sums of available strides, with $P-N=u$. Using compositions of available translations and truncated predecessors, $T_u=\rho_N\circ T_P$ and $\rho_u=\rho_P\circ T_N$. If $N=0$, the former is simply $T_P$.
\end{proof}

\begin{theorem}[Finite mixed rational classification]\label{thm:mixed}
Every nonempty finite set $Q\subset\mathbb Q_{>0}$ satisfies
\[
\boxed{\B_Q=\B_{1/L(Q)},\qquad\U_Q=\U_{1/L(Q)}.}
\]
\end{theorem}
\begin{proof}
By Theorem~\ref{thm:classification}, each single-stride algebra for a member of $Q$ is included in the canonical algebra at $u=1/L(Q)$; stage induction gives $\B_Q\subseteq\B_u$ and $\U_Q\subseteq\U_u$. Conversely Lemma~\ref{lem:bezout} constructs $\rho_u,T_u$ inside $\B_Q$. Fix $t\in Q$, so $t=pu$ for an integer $p\ge1$. All the hypotheses of Theorem~\ref{thm:block} hold in $\B_Q$, providing closure under bounded $u$-recursion and hence $\B_u\subseteq\B_Q$. The same compilation without a majorant proves the unbounded inclusion.
\end{proof}

For an infinite set, its entire algebra still contains only functions having finite programs. The appropriate denominator invariant is therefore a supernatural (Steinitz) number.
\begin{definition}[Supernatural denominator]
For nonempty $Q\subseteq\mathbb Q_{>0}$, define
\[
 \mathfrak L(Q)=\prod_{\ell\ \mathrm{prime}}\ell^{e_\ell(Q)},\qquad
 e_\ell(Q)=\sup_{s\in Q}v_\ell(\den(s))\in\NN_0\cup\{\infty\}.
\]
Write $\mathfrak L(Q)\mid\mathfrak L(P)$ if every exponent of the former is no larger than the corresponding exponent of the latter.
\end{definition}

\begin{theorem}[All rational stride sets]\label{thm:super}
For arbitrary nonempty $Q,P\subseteq\mathbb Q_{>0}$,
\[
\boxed{\B_Q\subseteq\B_P\iff\U_Q\subseteq\U_P
 \iff\mathfrak L(Q)\mid\mathfrak L(P).}
\]
Furthermore, for $\mathcal C\in\{\B,\U\}$,
\[
 \mathcal C_Q=\bigcup_{\substack{d\in\NN_{>0}\\d\mid\mathfrak L(Q)}}\mathcal C_{1/d}.
\]
Consequently, each rational-stride regime contains continuum many distinct algebras.
\end{theorem}
\begin{proof}
Each finite program uses only a finite $F\subset Q$ of strides. By Theorem~\ref{thm:mixed}, the algebra generated by $F$ equals $\mathcal C_{1/L(F)}$. The ordinary positive integers dividing $\mathfrak L(Q)$ are exactly those dividing some $L(F)$ for finite $F\subset Q$, since every integer has finitely many prime divisors. This proves the union formula. If supernatural divisibility holds, the union formula proves inclusion. Conversely, inclusion transfers each predecessor $\rho_{1/q}\in\mathcal C_s$ ($s\in Q$, reduced denominator $q$) to a finite-stage program using a finite subset of $P$. The obstruction Lemma~\ref{lem:pred} then gives $q\mid L(F)$ for some finite $F\subset P$, which forces each exponent inequality. There are continuum many distinct supernatural numbers, and each is realized by a choice of reciprocal prime-power strides.
\end{proof}

\section{Unbounded grid semantics and exact intersections}\label{sec:grid}
For $\delta=1/q$, let $I_0^\delta=\{0\}$ and $I_k^\delta=((k-1)\delta,k\delta]$ for $k\ge1$. An $m$-ary $\delta$-grid cell is a product $\prod_i I_{k_i}^\delta$. A \emph{descriptor} $(i,z)$ with $i=0$ or $1\le i\le m$, and $z\in\mathbb Z$, denotes the constant $z\delta$ when $i=0$ and $x_i+z\delta$ otherwise. We require its value to be nonnegative on the cell. Zero-coordinate occurrences are normalized to constant descriptors. A \emph{primitive-recursive grid code} is a primitive-recursive integer-valued function of the indices $\mathbf k$ producing such a valid descriptor on every cell.

\begin{lemma}[Uniform primitive-recursive code operations]\label{lem:uniformgrid}
Fix $\delta=1/q$ and arity $m$.  A signed grid descriptor is
encoded as a natural-number tuple specifying its tag
$i\in\{0,\ldots,m\}$ and signed offset $z\in\mathbb Z$.
There are primitive-recursive procedures for the following operations:
(i) finding the grid cell of the value of a valid descriptor on a given
input cell; (ii) substituting descriptors in a descriptor; and
(iii) evaluating the next descriptor in a finite-step recursion, given
the input cell, its predecessor depth, and the preceding-state descriptor.
All these procedures preserve nonnegativity on the input cell whenever
the represented program does.
\end{lemma}
\begin{proof}
Constants $z\delta\geq0$ belong to cell $z$ (with zero treated
separately).  For a nonzero input coordinate $x_i$ in cell $k_i\ge1$,
a valid expression $x_i+z\delta$ has image in cell $k_i+z$:
its range is $((k_i+z-1)\delta,(k_i+z)\delta]$, and validity
requires $k_i+z\ge1$.  An input coordinate fixed to zero is
normalized to a constant descriptor.  These cell-index operations
use integer addition, comparison and a sign tag, all primitive
recursive on the chosen encodings.

Substitution replaces the chosen coordinate by its inner descriptor
and adds the two signed offsets.  For a recursion on $y$ in cell
$k\ge1$, at bottom-up stage $j=1,\ldots,k$ the predecessor is
zero when $j=1$, and when $j>1$ it has descriptor
$y-(k-j+1)\delta$ and cell index $j-1$.  The transition code is
therefore evaluated on a tuple of explicitly calculated grid
indices, consisting of the frozen parameter cells, predecessor cell
and preceding-state cell.  The output descriptor is obtained by
substitution.  All operations have primitive-recursive encodings,
and primitive recursion over $j\le k$ computes the final descriptor.
The boundary convention $((k-1)\delta,k\delta]$ makes these
statements valid at endpoints, not merely on open interiors.
\end{proof}

\begin{lemma}[Syntax-uniform implementation of integer programs]\label{lem:pr-syntax}
Fix $\delta=1/q$. There is a computable map taking a standard
primitive-recursive syntax tree for $g:\mathbb N_0^m\to\mathbb N_0$
to a $\U_\delta$ syntax tree for a real function $\widetilde g$
agreeing with $g$ on every integer tuple. Parameters are frozen at
each recursion node, and the translation assumes no real addition
primitive. The depth clock $D_\delta$ belongs to $\U_\delta$.
\end{lemma}
\begin{proof}
Define $D_\delta$ by step recursion with base zero and transition
$S(u)$. Thus $D_\delta(N)=qN$ for integer $N$. A finite-state
step recursion computes $D_\delta(z)\bmod q$; for $q=1$ this is
the constant zero. Finite conditional selection is available by
unbounded step recursion.

Proceed by structural induction on primitive-recursive syntax.
Zero, successor and projections are inherited, and composition uses
the already translated children. For a primitive-recursion node
with data $a(\mathbf n)$ and $b(\mathbf n,k,u)$, define
\[
 G(\mathbf x,0)=\widetilde a(\mathbf x),\qquad
 G(\mathbf x,y)=
 \begin{cases}
 \widetilde b(\mathbf x,z,G(\mathbf x,z)),
       &D_\delta(z)\equiv0\pmod q,\\
 G(\mathbf x,z),&\text{otherwise},
 \end{cases}
 \quad z=\rho_\delta(y).
\]
All transition functions are constructed before this recursion.
On integer recursion input $N$, the bottom-up predecessor addresses
are $0,\delta,\ldots,N-\delta$. The gate selects exactly
$0,1,\ldots,N-1$, in that order. At each selected stage the
parameters, selected address and stored state are integers, so the
induction hypothesis for $\widetilde b$ applies. The other stages
retain the state. Hence $G(\mathbf n,N)$ equals the ordinary
primitive-recursion value, including at $N=0$.

This performs an effective finite collection of syntax operations at
each source node, proving uniformity. Integer pairing, sign tags,
comparisons and finite cases are translated as primitive-recursive
integer subprograms, not added as real primitives.
\end{proof}

\label{sec:proofcheck}
\begin{example}[A complete integer microstep schedule]
\label{ex:integertrace}
Take $\delta=1/2$ and an ordinary primitive recursion
$f(n,0)=a(n)$ and $f(n,k+1)=b(n,k,f(n,k))$.
On the integer input $k=3$, the $\delta$-recursion has six
bottom-up microsteps; their predecessor addresses, in order, are
\[
 0,\ \tfrac12,\ 1,\ \tfrac32,\ 2,\ \tfrac52.
\]
The gate $D_{1/2}(z)\equiv0\pmod2$ is active exactly at
$z=0,1,2$. The three active updates are therefore
\[
 b(n,0,a(n)),\quad
 b(n,1,f(n,1)),\quad
 b(n,2,f(n,2)),
\]
while the other three microsteps leave the stored value unchanged.
This verifies the indexing, including the first update at predecessor
zero and the right-closed input endpoint; it does not assume integer
addition or multiplication as real-valued basis operations.
\end{example}

\begin{corollary}[Variable shifts]\label{lem:microformal}
There are $\U_\delta$ functions which, for real $x\ge0$ and
integer $n\ge0$, return $x+n\delta$ and
$\rho_\delta^{[n]}(x)$. There is also a function
$\operatorname{AddInt}(x,n)=x+n$ on these inputs.
\end{corollary}
\begin{proof}
Use the gated recursion of Lemma~\ref{lem:pr-syntax} on integer
counter $n$, with base the frozen real parameter $x$. At each
selected stage apply, respectively, $T_\delta$, $\rho_\delta$,
or $S$ to the state. There are exactly $n$ selected stages.
Here $\rho_\delta$ is a one-step definable function and
$T_\delta=\rho_\delta^{[q-1]}\circ S$, so all transitions are
available before the recursion. The state may be real because
these transitions are exact on all nonnegative real inputs.
\end{proof}

\begin{lemma}[Uniform decoder for affine grid descriptors]\label{lem:grid-decode}
Let $c(\mathbf k)$ be a primitive-recursive code returning a valid
descriptor $(i,z)$ on each $\delta$-product cell. Its extensional
real function is definable in $\U_\delta$ by a uniform syntax
construction.
\end{lemma}
\begin{proof}
Compute the integer depths $k_j=D_\delta(x_j)$ and then the
integer tag and signed offset using Lemma~\ref{lem:pr-syntax}.
Build a base $v=0$ if $i=0$ and $v=x_i$ otherwise, using a
finite unbounded conditional on the tag. If $z\ge0$, run the
variable-length shift $T_\delta^{[z]}(v)$; if $z<0$, run the
variable-length descent $\rho_\delta^{[-z]}(v)$, as in
Corollary~\ref{lem:microformal}. For a nonnegative valid constant
$z\delta$ this yields exactly $z\delta$. For a variable descriptor,
validity on the whole current cell ensures
$v+z\delta\ge0$ there, and hence the truncation in a descent does
not change the intended affine expression. All tag and sign tests
act on integer outputs of the lifted code, while the real coordinate
$v$ is frozen during each shift. Thus the resulting program agrees
with the specified descriptor on every cell, including singleton
zero cells and right endpoints.\end{proof}

\begin{lemma}[Coding direction]\label{lem:coding}
Every function of $\U_{1/q}$ has a primitive-recursive $1/q$-grid code.
\end{lemma}
\begin{proof}
The base functions have evident codes. For composition, the image of a grid cell under a nonnegative constant or translated-coordinate expression lies wholly inside a unique grid interval, whose index is obtained from the input index by a fixed integer shift or constant (Lemma~\ref{lem:uniformgrid}). Substituting the corresponding outer descriptors is primitive recursive. For recursion, the last index $k$ gives precisely $k$ steps. In bottom-up order the predecessor address at stage $j$ lies in the $j$th grid interval ($j=0,\ldots,k-1$), with descriptor zero at $j=0$ and $y-(k-j)\delta$ otherwise. Evaluate the known code for the transition on the preceding descriptor and the accumulated state descriptor. The exact index calculation and tagged-offset computation are in Lemma~\ref{lem:uniformgrid}; thus this is a primitive recursion on the integer stage count with finite tuples of integer tags and signed offsets. The resulting descriptor is valid on the whole original cell. Inducting over program stages proves the claim.
\end{proof}

\begin{lemma}[Realization direction]\label{lem:realization}
Every nonnegative finite-arity function with a primitive-recursive $1/q$-grid code belongs to $\U_{1/q}$.
\end{lemma}
\begin{proof}
By Lemma~\ref{lem:pr-syntax}, the code can be evaluated on the
integer grid indices without adding arithmetic to the real basis.
Lemma~\ref{lem:grid-decode} then constructs a real-valued program
realizing the decoded descriptor on every cell. Its stage order is:
first the depth clocks, next the lifted integer code, then the tag
and sign guards, and finally the variable translations or descents.
All recursion nodes in this construction are unbounded, so no
unverified intermediate majorant is invoked.
\end{proof}

\begin{theorem}[Exact unbounded grid semantics]\label{thm:grid}
A nonnegative finite-arity $f$ belongs to $\U_{1/q}$ if and only if it has a primitive-recursive $1/q$-grid code.
\end{theorem}
\begin{proof}Lemmas~\ref{lem:coding} and~\ref{lem:realization}.\end{proof}

\begin{lemma}[Connectivity of the rational overlap graph]\label{lem:overlapgraph}
Fix positive integers $q,v$, $d=\gcd(q,v)$ and
$\Delta=1/d$.  Within an interval $J=((K-1)\Delta,K\Delta]$
with $K\ge1$, consider the $1/q$-cells and $1/v$-cells meeting
its interior.  Make a bipartite graph whose vertices are these
cells and whose edges are pairs with positive-length overlap.
This graph is connected.  More generally, the overlap graph of
product cells inside a product of such $J$'s is connected.
\end{lemma}
\begin{proof}
Write $q=da$ and $v=db$ with $\gcd(a,b)=1$.
An internal common grid boundary in $J$ would satisfy
$j/q=\ell/v$ and be strictly between $(K-1)/d$ and $K/d$.
Equivalently, $jb=\ell a$, so $a\mid j$ and the boundary
belongs to $\Delta\mathbb Z$; there is no such point inside $J$.
Move from left to right across the ordered union of all fine-grid
boundaries.  At a boundary strictly inside $J$ only one family
changes cells, while the cell from the other family overlaps both
adjacent cells in intervals of positive length.  This connects
successive vertices and the entire graph.  For product cells, consider first two vertices from the $1/q$
family whose cell indices differ in just one coordinate. In that
coordinate, a path in the one-dimensional bipartite graph connects
the two $1/q$ intervals. In every other coordinate keep its $1/q$
interval fixed, and choose a $1/v$ interval with positive-length
overlap with it. Each intermediate $1/v$ interval on the path,
together with those fixed choices, gives a $1/v$ product cell
overlapping both neighboring $1/q$ product cells in positive
relative volume. Thus the one-dimensional path lifts to the product
graph. Changing one coordinate at a time connects all $1/q$
product vertices. Every $1/v$ product vertex overlaps some $1/q$
product vertex, so the full graph is connected. Coordinates fixed
at zero are omitted throughout, with overlaps taken in the
remaining relative domain. If all coordinates are zero, there is
only one cell from each family and they overlap in that singleton.
\end{proof}

Figure~\ref{fig:overlap-grids} shows the one-dimensional overlap
mechanism for denominators two and three.
\begin{figure}[htbp]
\centering
\begin{tikzpicture}[x=1cm,y=1cm,
  every node/.style={font=\small},
  cell/.style={draw=black!65,rounded corners=2pt,
    minimum width=10mm,minimum height=6mm,inner sep=3pt}]
  \node[anchor=east,font=\footnotesize] at (-.3,3.75) {$x$};
  \draw[black!65] (0,3.75) -- (6,3.75);
  \foreach \xx/\tt in {0/0,2/{\frac13},3/{\frac12},4/{\frac23},6/1} {
    \draw (\xx,3.66) -- (\xx,3.84);
    \node[above] at (\xx,3.9) {$\tt$};
  }
  \node[anchor=east] at (-.3,2.95) {Half-grid};
  \filldraw[fill=blue!6,draw=black!65] (0,2.6) rectangle (3,3.3);
  \filldraw[fill=blue!6,draw=black!65] (3,2.6) rectangle (6,3.3);
  \node at (1.5,2.95) {$A_1=(0,\tfrac12]$};
  \node at (4.5,2.95) {$A_2=(\tfrac12,1]$};
  \node[anchor=east] at (-.3,1.95) {Third-grid};
  \filldraw[fill=black!4,draw=black!65] (0,1.6) rectangle (2,2.3);
  \filldraw[fill=black!4,draw=black!65] (2,1.6) rectangle (4,2.3);
  \filldraw[fill=black!4,draw=black!65] (4,1.6) rectangle (6,2.3);
  \node at (1,1.95) {$B_1=(0,\tfrac13]$};
  \node at (3,1.95) {$B_2=(\tfrac13,\tfrac23]$};
  \node at (5,1.95) {$B_3=(\tfrac23,1]$};
  \node[anchor=east,align=right,font=\footnotesize] at (-.3,.1)
    {Positive-length\\overlap graph};
  \node[cell,fill=blue!6] (a1) at (1.5,.45) {$A_1$};
  \node[cell,fill=blue!6] (a2) at (4.5,.45) {$A_2$};
  \node[cell,fill=black!4] (b1) at (1,-.75) {$B_1$};
  \node[cell,fill=black!4] (b2) at (3,-.75) {$B_2$};
  \node[cell,fill=black!4] (b3) at (5,-.75) {$B_3$};
  \foreach \a/\b in {a1/b1,a1/b2,a2/b2,a2/b3}
    \draw[semithick] (\a) -- (\b);
  \draw[black!25] (6.65,-1.15) -- (6.65,4.35);
  \node[anchor=north west,align=left,text width=4.2cm] at (7.15,3.9)
    {No shared boundary in $(0,1)$.\\[8pt]
     Connected overlaps force one affine descriptor on $(0,1]$.};
  \node[anchor=west,align=left] at (7.15,.5)
    {$\displaystyle\tfrac12\mathbb Z\cap\tfrac13\mathbb Z=\mathbb Z$\\[7pt]
     $\displaystyle\mathcal U_{1/2}\cap\mathcal U_{1/3}=\mathcal U_1$};
\end{tikzpicture}
\caption{The overlap mechanism behind literal unbounded intersections,
illustrated on the common coarse cell $(0,1]$. Edges correspond to
positive-length overlaps; all intervals are left-open and right-closed.
Connectedness forces matching affine expressions, and the common offset
is integral. The cell at zero is treated separately. For arbitrary
denominators $q,v$, the common coarse cell has width $1/\gcd(q,v)$;
the effective realization argument is supplied by the grid semantics.}
\label{fig:overlap-grids}
\end{figure}
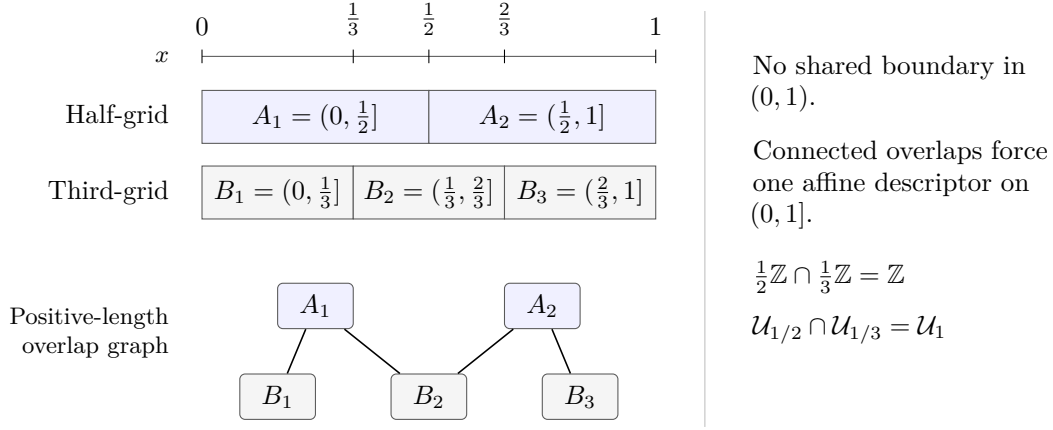

\begin{lemma}[Gluing two rational grids]\label{lem:glue}
Let $d=\gcdop(q,v)$ and $\Delta=1/d$. If $f$ has a $1/q$-grid code and a $1/v$-grid code, it is represented by a single descriptor on each $\Delta$-grid cell. Its offset belongs to $\Delta\mathbb Z$. If either fine code is primitive recursive, so is the resulting coarse descriptor code.
\end{lemma}
\begin{proof}
Inside each interval of length $\Delta$, the interior boundaries of the two fine grids never coincide: a common boundary would be a positive multiple of both $1/q$ and $1/v$, hence of $\Delta$. The overlap graph of their fine cells is connected by Lemma~\ref{lem:overlapgraph}. On an open overlap, distinct coordinate-affine expressions cannot agree, so the descriptors must match as real affine functions. Cartesian products of these connected graphs treat any number of coordinates; singleton zero-coordinate cells are handled by normalization. Endpoint values follow because intervals are right-closed.

The offset is simultaneously in $(1/q)\mathbb Z$ and $(1/v)\mathbb Z$, and their intersection equals $(1/d)\mathbb Z$. To extract its code, for each nonzero coarse index $K_i$ sample the first $1/q$-grid index $(K_i-1)q/d+1$ inside that coarse cell, keeping zero indices zero. The selected fine code is correct throughout the coarse cell. Its offset in $1/q$ units is divisible by $q/d$, and fixed-integer division gives the $1/d$-offset. The extraction is primitive recursive.
\end{proof}

\begin{theorem}[Literal unbounded intersections]\label{thm:meet}
For all $q,v\ge1$,
\[\boxed{\U_{1/q}\cap\U_{1/v}=\U_{1/\gcd(q,v)}.}\]
For nonempty rational stride sets $Q,P$, the same identity holds with the right side replaced by $\U_R$, where $\mathfrak L(R)$ is the primewise greatest common divisor of $\mathfrak L(Q)$ and $\mathfrak L(P)$.
\end{theorem}
\begin{proof}
By Theorem~\ref{thm:grid}, membership in both single-stride algebras supplies two primitive-recursive fine-grid codes. Lemma~\ref{lem:glue} provides a primitive-recursive coarse-grid code; Theorem~\ref{thm:grid} realizes it in $\U_{1/d}$. The reverse inclusion follows from Theorem~\ref{thm:classification}.

For arbitrary sets, every member of the intersection uses some finite $Q_0\subset Q$ and $P_0\subset P$. Finite mixed normalization reduces these to two ordinary denominators; the first part places the function in their greatest-common-divisor class. That denominator divides the supernatural gcd, hence Theorem~\ref{thm:super} places it in $\U_R$. Conversely, each finite denominator dividing that supernatural gcd is realized inside both stride families, giving inclusion of $\U_R$ into the intersection.
\end{proof}

\section{Bounded intersections and integer predicate precision}
\label{sec:bounded-meets}
For positive integers $q,v$ put $d=\gcdop(q,v)$. The remaining literal
bounded-intersection assertion is
\begin{equation}\label{eq:target}
 \B_{1/q}\cap\B_{1/v}=\B_{1/d}.
\end{equation}
Theorem~\ref{thm:meet} places every common function in $\U_{1/d}$,
with a normalized affine descriptor on each coarse product cell.
When one denominator divides the other, \eqref{eq:target} follows from
the inclusion theorem; otherwise both fine strides are smaller than
$\delta=1/d$. We reduce the latter case to common Boolean predicates,
then identify their integer precision classes. The dependencies are the
rational compilers, Lemma~\ref{lem:onecoord} and the unbounded grid theorem;
no bounded-intersection identity is assumed.

\subsection{Globally bounded grid tools}
Throughout this section $\delta=1/d$ and all tools belong to $\B_\delta$.
The fixed translation $T_\delta$ and predecessor $\rho_\delta$ are available
by the rational toolkit. Write $\nz(t)=[t>0]$ and $\zz(t)=[t=0]$.
Both tests are boundedly definable with constant bound one.

\begin{lemma}[Round-up, grid subtraction and residual]\label{lem:tools}
The following functions belong to $\B_\delta$:
\begin{align*}
 \C_\delta(y)&=\delta\lceil y/\delta\rceil,\\
 R_\delta(a,b)&=\rho_\delta^{[D_\delta(b)]}(a),\\
 r_\delta(x)&=R_\delta(x,\rho_\delta(\C_\delta(x))).
\end{align*}
Their bounds are respectively $y+\delta$, $a$, and $\delta$ for the last
function's actual values. On grid arguments $a,b\in\delta\N$,
$R_\delta(a,b)=(a-b)_+$. For $x>0$,
$r_\delta(x)=x-\C_\delta(x)+\delta\in(0,\delta]$, while $r_\delta(0)=0$.
\end{lemma}
\begin{proof}
Use $\delta$-recursion with base zero and transition $T_\delta(u)$ for
$\C_\delta$. Its solution is $\delta D_\delta(y)\leq y+\delta$, and
$T_\delta(y)$ is an earlier global majorant. For $R_\delta$ use base $a$
and transition $\rho_\delta(u)$ with the frozen earlier bound $a$.
Composition gives $r_\delta$ and its displayed identities. We do not
admit a new recursion under a diagonal-only bound: $R_\delta$ is first
constructed on all $(a,b)$ with bound $a$, then composed.
\end{proof}

\begin{lemma}[Transfer of the within-cell coordinate]\label{lem:transfer}
There is $J_\delta(r,y)\in\B_\delta$ with the earlier global majorant
$y+\delta$, such that for $0<r\leq\delta$ and $y\in\delta\mathbb N_{>0}$,
\[
 J_\delta(r,y)=r+y-\delta.
\]
\end{lemma}
\begin{proof}
First zero-guard $r$ outside $[0,\delta]$, using the fixed threshold
$\nz(\rho_\delta(r))$ and earlier bound $r$. Call the result $g(r)$;
it has actual global bound $\delta$. Define
\[
 J_\delta(r,0)=0,\qquad
 J_\delta(r,y)=
 \begin{cases}
 g(r),&\rho_\delta(y)=0,\\
 J_\delta(r,\rho_\delta(y))+\delta,&\rho_\delta(y)>0.
 \end{cases}
\]
The transition selector is an earlier bounded function: its two candidate
values $g(r)$ and $u+\delta$ are both bounded by $u+\delta$ on all
$(r,z,u)$. The resulting recursion is
$g(r)+(D_\delta(y)-1)\delta$ for $y>0$, and is globally bounded by
$y+\delta$, independently of $r$. This earlier majorant admits the
recursion. The identity follows on the stated grid inputs.
\end{proof}

\begin{lemma}[Bounded counting of an earlier predicate]\label{lem:count}
For every Boolean $P(\mathbf x,t)\in\B_\delta$, there is
$K_P(\mathbf x,y)\in\B_\delta$, with earlier global bound $y+\delta$,
such that on $y=n\delta$,
\[
 K_P(\mathbf x,n\delta)=\delta\sum_{k=1}^{n}P(\mathbf x,k\delta).
\]
\end{lemma}
\begin{proof}
Take base zero, and at predecessor address $z$ increment the state by
$\delta$ exactly when $P(\mathbf x,\C_\delta(z)+\delta)=1$.
Both transition branches are bounded by $u+\delta$, so the transition
selector is an earlier bounded function. On every real input there are
$D_\delta(y)$ transitions and at most that many increments; hence the
new recursion has the earlier bound $y+\delta$. On grid inputs the tested
thresholds are exactly $\delta,2\delta,\ldots,n\delta$.
\end{proof}

\subsection{Reduction to Boolean intersections}
\begin{theorem}[Boolean predicates suffice]\label{thm:booleanreduce}
Fix $q,v$ and $d=\gcd(q,v)$. If every Boolean function of every finite arity
in $\B_{1/q}\cap\B_{1/v}$ belongs to $\B_{1/d}$, then
\eqref{eq:target} holds for all functions of every finite arity.
The reverse implication is immediate.
\end{theorem}
\begin{proof}
The nested case is already known. Otherwise put $\delta=1/d$,
$h_q=1/q<\delta$, $h_v=1/v<\delta$, and let $f$ be an $m$-ary common
bounded function. Its coarse grid semantics gives a normalized affine
descriptor on every $\delta$-product cell. Put
\[
 \mathbf e(\mathbf x)=(\C_\delta(x_1),\ldots,\C_\delta(x_m)),
 \qquad E(\mathbf x)=f(\mathbf e(\mathbf x)).
\]
The endpoint value $E$ is $\delta$-grid-valued and belongs to both source
algebras. A one-coordinate bound of $f$, say $f\leq x_j+C$ with integer
$C\geq0$, gives
\begin{equation}\label{eq:endbound}
 E(\mathbf x)\leq B(\mathbf x):=\C_\delta(x_j)+C.
\end{equation}
The grid-valued bound $B$ is already in $\B_\delta$.

\paragraph{Recover the endpoint value.}
The Boolean function
\[
 P(\mathbf x,t)=\zz\bigl(R_\delta(\C_\delta(t),E(\mathbf x))\bigr)
\]
belongs to both source algebras and equals
$[E(\mathbf x)\geq\C_\delta(t)]$, because both values being compared
are on the grid. By the assumed Boolean intersection, $P\in\B_\delta$.
Lemma~\ref{lem:count} and \eqref{eq:endbound} now yield
\[
 E(\mathbf x)=K_P(\mathbf x,B(\mathbf x))\in\B_\delta.
\]
This counts the true positive grid thresholds up to $E$. No unbounded
grid descriptor decoder is used.

\paragraph{Recover the coordinate tag.}
For each $i$, replace the $i$th coordinate of $\mathbf e(\mathbf x)$ by
$\rho_{h_q}(\C_\delta(x_i))$; call this tuple $\mathbf e^{q,i}(\mathbf x)$.
Define, in $\B_{1/q}$,
\[
 T_i^q(\mathbf x)=\nz\bigl(R_{h_q}(E(\mathbf x),
                         f(\mathbf e^{q,i}(\mathbf x)))\bigr).
\]
If $x_i=0$, the two tuples coincide and this is zero. Otherwise the
probe is strictly inside the same coarse cell, since $h_q<\delta$.
If that cell's descriptor uses coordinate $i$, its value decreases by
$h_q$ at the probe; for every other normalized descriptor it is unchanged.
Both probe values lie in $h_q\N$, so grid subtraction is exact.
Consequently $T_i^q$ is precisely the indicator of normalized tag $i$.
The same construction with $h_v$ gives $T_i^v$ in $\B_{1/v}$ with
identical values on all real tuples. Thus the intrinsic tag predicate
$T_i=T_i^q=T_i^v$ lies in the source intersection, and by assumption
belongs to $\B_\delta$.

\paragraph{Reconstruct the real output.}
If the descriptor is constant, then $f=E$. If its tag is $i>0$, then
$x_i>0$ and validity of the affine descriptor implies $E\geq\delta$;
otherwise the descriptor would be negative at interior points of its cell.
In this case
\[
 f(\mathbf x)=J_\delta(r_\delta(x_i),E(\mathbf x)),
\]
by Lemmas~\ref{lem:tools} and~\ref{lem:transfer}. All candidate functions
$E$ and $J_\delta(r_\delta(x_i),E)$ are globally bounded by the already
constructed $E+\delta$. Finite bounded selectors controlled by the disjoint
tag predicates therefore reconstruct $f$ in $\B_\delta$. This proves the
nontrivial inclusion, with global majorants at every auxiliary stage.
\end{proof}

\begin{corollary}[Transfer of counterexamples]\label{cor:witness}
If an $m$-ary function violates \eqref{eq:target}, then some Boolean
function of arity at most $m+1$ violates the same identity. There is also
a uniform construction of one Boolean counterexample of arity $m+2$
from any such common-function witness and its two valid source descriptions.
\end{corollary}
\begin{proof}
At least one of the $P,T_1,\ldots,T_m$ in the proof must fail coarse
bounded membership, since otherwise the reconstruction succeeds.
For a single uniform witness, use an extra finite index argument to select
$P(\mathbf x,t)$ or $T_i(\mathbf x)$, with all selectors bounded by one.
If this combined predicate belonged to the coarse algebra, restrictions
of its index to fixed integer constants would put every selected predicate
there, a contradiction. Construction is under the promise that the given
source descriptions define the same function; it does not decide that
promise or identify which separate predicate is a counterexample.
\end{proof}

\subsection{Integer algebras and stretched predicates}
\begin{definition}[Integer levels used here]
On $\N$ define
\[
 \E_0=[0,S,\mathrm{proj};\mathrm{comp},\mathrm{BR}],\qquad
 \E_1=[0,S,\mathrm{proj},+;\mathrm{comp},\mathrm{BR}].
\]
Here $\mathrm{BR}$ is scalar ordinary primitive recursion with an earlier
global majorant from the same algebra. Let $\E_{i,*}$ denote the
Boolean-valued functions of all positive finite arities in $\E_i$.
These explicit definitions fix the numbering, matching the corresponding
successor and addition levels in Clote's exposition
\cite[Definitions 3.28--3.29]{Clote1999}. The integer class $\E_0$
is the ordinary bounded-recursion class denoted $\mathcal E^0$ in
the companion minimal-basis articles. This identification concerns
integer functions; restriction and round-up lifting relate it to the
real algebra $\B_1$ as specified below.
\end{definition}

For an $m$-ary Boolean predicate $p$ on integers put
\[
 \widehat p(\mathbf x)=p(\lceil x_1\rceil,\ldots,\lceil x_m\rceil),
 \qquad
 (S_a p)(\mathbf n)=p(\lceil n_1/a\rceil,\ldots,\lceil n_m/a\rceil).
\]
The letter $S_a$ here denotes stretching of a predicate, not numerical
successor. Put $\mathcal P_a=\{p:S_a p\in\E_{0,*}\}$.
All arities are included in this notation.

\begin{lemma}[Restriction, lifting and change of scale]\label{lem:scale}
\begin{enumerate}[label=(\roman*),leftmargin=2em]
\item Restricting a $\B_1$ function to integer inputs gives a function in
$\E_0$. For each $g\in\E_0$, its round-up lift
$\widehat g(\mathbf x)=g(\lceil\mathbf x\rceil)$ belongs to $\B_1$.
\item For an integer $a\geq1$, the meta-level transformation
\[
 (\mathfrak T_a f)(\mathbf x)=a f(\mathbf x/a)
\]
maps $\B_{1/a}$ onto $\B_1$, preserving construction order and majorant
inequalities. More generally, if $d\mid q$, $\mathfrak T_d$ maps
$\B_{1/q}$ onto $\B_{1/(q/d)}$, and does the same for the unbounded algebras.
\end{enumerate}
\end{lemma}
\begin{proof}
On integers, unit-stride recursion is ordinary primitive recursion and its
majorant remains valid. For lifting, first construct $\C_1=\lceil x\rceil$
using Lemma~\ref{lem:tools}. Lift source primitives by rounding their
inputs. Lift composition by composition: inner lifted values are integers.
For a source recursion, real unit-stride recursion performs exactly
$\lceil y\rceil$ steps, and the lifted transition sees the appropriate
integer predecessor $\lceil z\rceil$ and integer state. Its solution is
$g(\lceil\mathbf x\rceil,\lceil y\rceil)$, bounded on every real input
by the earlier lift of the source majorant. This proves (i).

For (ii), conjugation commutes with composition and turns stride $1/q$
into $d/q$ while multiplying every majorant inequality by $d$. Unit
successor becomes translation by the integer $d$. The image has the same
recursion operation as its claimed target, and also contains unit
successor: apply $a(d-1)$ predecessors of stride $1/a$ to translation
by $d$, where $a=q/d$. Conversely the inverse transformation sends target
unit successor to translation by $1/d$, available in $\B_{1/q}$ since
$d\mid q$. Zero and projections are unchanged. The same syntax induction
works without majorants. These are transformations between programs,
not assertions that multiplication or division of arbitrary real inputs
is an internal primitive of either algebra.
\end{proof}

\begin{theorem}[Exact Boolean precision dictionary]\label{thm:dictionary}
For every integer predicate $p$ and every $a\geq1$,
\[
 \widehat p\in\B_{1/a}\quad\Longleftrightarrow\quad
 S_a p\in\E_{0,*}\quad\Longleftrightarrow\quad p\in\mathcal P_a.
\]
In particular $\mathcal P_1=\E_{0,*}$.
\end{theorem}
\begin{proof}
If $\widehat p\in\B_{1/a}$, conjugate it by $\mathfrak T_a$ and apply
the nonzero test. This gives the Boolean $\B_1$ function
$p(\lceil\mathbf x/a\rceil)$. Its integer restriction is $S_a p$.

Conversely lift $S_a p$ from $\E_0$ to $\B_1$, conjugate by
$\mathfrak T_a^{-1}$, and apply the nonzero test in $\B_{1/a}$.
The resulting function is
$p(\lceil\lceil a\mathbf x\rceil/a\rceil)=\widehat p(\mathbf x)$,
using the coordinatewise identity
$\lceil\lceil at\rceil/a\rceil=\lceil t\rceil$ for integer $a\geq1$.
The inverse conjugate itself has values zero or $1/a$, so the nonzero
test is essential to restore Boolean values.
\end{proof}

\paragraph{Example: parity under stretching.}
Let $p(n)=[n\text{ is odd}]$. A two-state recursion with earlier
constant bound one puts $p$ in $\E_{0,*}$. For $a=2$,
\[
 \bigl(S_2p(n)\bigr)_{n=0}^{6}=(0,1,1,0,0,1,1),\qquad
 \widehat p(x)=p(\lceil x\rceil).
\]
Thus $\widehat p$ is one on $(0,1],(2,3],\ldots$ and zero on
$\{0\},(1,2],(3,4],\ldots$. Conjugation gives
$\mathfrak T_2\widehat p(x)=2p(\lceil x/2\rceil)$; the nonzero
test restores Boolean values, and integer restriction gives $S_2p$.
Conversely, lifting $S_2p$, applying $\mathfrak T_2^{-1}$ and then
the nonzero test gives
$p(\lceil\lceil2x\rceil/2\rceil)=\widehat p(x)$.
The later numerical monotonicity concerns these coarse-cell predicates;
the full algebras $\B_{1/2}$ and $\B_{1/3}$ remain incomparable.

\begin{lemma}[Contractions and monotonicity in precision]\label{lem:monotone}
For integers $1\leq a\leq b$, the integer function
$c_{a,b}(n)=\lceil an/b\rceil$ belongs to $\E_0$, and
\[
 \mathcal P_a\subseteq\mathcal P_b.
\]
This is monotonicity in the ordinary numerical order, for these coarse-grid
Boolean precision classes; it is not an inclusion claim for the full real
algebras.
\end{lemma}
\begin{proof}
The successive differences of $c_{a,b}$ are zero or one and are periodic
with fixed period $b$. A finite residue clock, bounded by a fixed integer,
controls a recursion that starts at zero and increments exactly at those
residues. Its values are at most $n$, so the earlier majorant $n$ admits
the recursion. All finite-state selectors have earlier constant or
state-plus-one bounds. Hence $c_{a,b}\in\E_0$.

For integer tuples,
\[
 S_b p=(S_a p)\circ(c_{a,b},\ldots,c_{a,b}),
\]
because $\lceil\lceil an/b\rceil/a\rceil=\lceil n/b\rceil$.
Closure under composition proves the inclusion.
\end{proof}

\subsection{The union of precision classes is the addition level}
\begin{lemma}[Linear intermediate-value envelope]\label{lem:envelope}
For every fixed valid $\E_1$ program $P$, there is an effectively
extractable integer $K_P\geq1$ such that its chosen bottom-up evaluator,
on integer inputs at most $N$, materializes only values at most
$K_P(N+1)$. Extraction is syntactic; correctness is conditional on validity
of its majorant annotations.
\end{lemma}
\begin{proof}
Fix the bottom-up evaluator which computes inner arguments before an
outer call and computes a recursion's states from address zero upwards.
Majorant annotations certify admission; their numerical evaluation is not
needed to run the program. Simultaneously extract integers $L_P,K_P\geq1$
such that, for inputs at most $N$, the output is at most $L_P(N+1)$
and every materialized value is at most $K_P(N+1)$. Take
$(L_P,K_P)=(1,1)$ for zero and projections and $(2,2)$ for successor
and addition.

For $P=g(h_1,\ldots,h_r)$ put $L=\max_i L_{h_i}$ and take
\[
 L_P=L_g(L+1),\qquad
 K_P=\max\{1,K_{h_1},\ldots,K_{h_r},K_g(L+1)\}.
\]
Indeed, the arguments of $g$ are at most $L(N+1)$, and its envelope
therefore contributes at most $K_g(L+1)(N+1)$.

At a bounded-recursion node $F$ with earlier base $g$, transition $b$
and majorant $M$, set $C=\max\{1,L_M\}$ and take
\[
 L_F=L_M,\qquad
 K_F=\max\{1,L_M,K_g,K_M,K_b(C+1)\}.
\]
Every prefix state $F(\mathbf x,j)$ with $j\leq N$ satisfies
$F(\mathbf x,j)\leq M(\mathbf x,j)\leq L_M(N+1)$.
Hence every call to $b$ receives frozen parameters, predecessor address
and previous state all at most $C(N+1)$; its evaluator materializes
only values at most $K_b(C+1)(N+1)$. The base is covered by $K_g$.
The term $K_M$ also covers evaluating a majorant on the original
parameters and any prefix address, if desired. Repeating these calls
changes their number, not this value bound. These finite recurrences
are syntactic and computable; their correctness uses the validity of
the source majorant annotations. Structural induction proves the claim.
\end{proof}

\begin{lemma}[Clipped scalar simulation on a fine grid]\label{lem:clip}
Fix a valid scalar $\E_1$ program $P$ and choose an integer $a\geq K_P$.
Put $h=1/a$. There is a $\B_h$ program $\widetilde P(\mathbf u,w)$
with global bound $\C_h(w+1)$ such that, for integer $n_i\leq N$,
\[
 \widetilde P(\mathbf n/a,N)=P(\mathbf n)/a.
\]
The parameter $w$ is frozen at every auxiliary recursion. The simulation
uses scalar recursion only, with no pairing or simultaneous-recursion
assumption.
\end{lemma}
\begin{proof}
Set $A(w)=\C_h(w+1)$, an earlier grid-valued resource bound. Grid minimum
is available by composition:
\[
 \min_h(u,v)=R_h(\C_h(u),R_h(\C_h(u),\C_h(v))).
\]
It equals $\min(\C_h(u),\C_h(v))$ on all real inputs. Define clipped grid
addition by $h$-recursion on its second input, with base
$\min_h(u,A(w))$ and transition $\min_h(T_h(s),A(w))$.
Its earlier global majorant is $A(w)$, independently of both inputs.
On grid arguments it returns $\min(u+v,A(w))$.

Translate source zero to zero; source projection to
$\min_h(u_i,A(w))$; source successor to
$\min_h(T_h(\C_h(u)),A(w))$; and source addition to clipped grid addition.
Every primitive output lies on the $h$ grid and is globally bounded by
$A(w)$. Translate composition by composition with $w$ frozen. Translate
an ordinary source scalar recursion to $h$-recursion on the scaled
recursion argument, using the already translated base and transition and
the earlier majorant $A(w)$. Explicitly, if $F$ has source base $g$ and transition $b$, define
\[
 \widetilde F(\mathbf u,0,w)=\widetilde g(\mathbf u,w),\qquad
 \widetilde F(\mathbf u,y,w)=\widetilde b\bigl(
 \mathbf u,\rho_h(y),\widetilde F(\mathbf u,\rho_h(y),w),w\bigr)
 \quad(y>0).
\]
Permuting arguments places $y$ last for the recursion rule; $w$ remains
frozen. Since both translated functions already return values at most
$A(w)$ on their full domains, induction on the finite $h$-descent proves
this bound before using the source simulation identity.
Its state is always clipped, on every real
auxiliary input, not just on the intended integer-simulation diagonal.
Thus every recursion is admissible before any correctness argument.

On intended integer inputs with $w=N$, the source envelope gives
$t\leq K_P(N+1)\leq a(N+1)$ for every value $t$ in the complete source
evaluation. Hence its scaled value $t/a$ is on the $h$ grid and at most
$N+1=A(N)$. No clipping changes any such value. A source recursion on
integer $k$ becomes an $h$-recursion on $k/a$, with exactly $k$ steps;
its predecessor and accumulated state are the scaled source values.
Induction on this finite evaluation proves the claimed identity.
\end{proof}

\begin{lemma}[An internal input contraction]\label{lem:input}
For every integer $a\geq1$, $I_a(x)=\lceil x\rceil/a$ belongs to
$\B_{1/a}$, with earlier global majorant $x+1/a$.
\end{lemma}
\begin{proof}
Put $h=1/a$. Use the earlier finite-valued clock $D_h(z)\bmod a$ in an
$h$-recursion with base zero, adding $h$ exactly when that clock is zero
at the predecessor address. The transition selector is bounded by
$u+h$. There are $D_h(x)$ microsteps, and the selected indices are
$0,a,2a,\ldots$. The output is therefore
$h\lceil D_h(x)/a\rceil=h\lceil x\rceil\leq x+h$.
The earlier translation $T_h(x)$ admits the new recursion on its entire
real input domain.
\end{proof}

\begin{theorem}[Exact union of coarse Boolean precision classes]
\label{thm:union}
\[
 \bigcup_{a\geq1}\mathcal P_a=\E_{1,*}.
\]
\end{theorem}
\begin{proof}
If $S_a p\in\E_0$, then
$p(\mathbf n)=(S_a p)(a\mathbf n)$ belongs to $\E_1$, where fixed
integer multiplication of an argument is available by addition.
This proves one inclusion.

Conversely take a Boolean $\E_1$ program $P$ for $p$, choose
$a\geq K_P$, and use Lemma~\ref{lem:clip}. For each input coordinate $j$
form in $\B_{1/a}$ the Boolean function
\[
 F_j(\mathbf x)=\nz\bigl(\widetilde P(
                 I_a(x_1),\ldots,I_a(x_m),\C_1(x_j))\bigr).
\]
The predicate $G_j=[\C_1(x_i)\leq\C_1(x_j)\text{ for all }i]$ belongs
to $\B_1$, using grid subtraction and finite Boolean conjunction.
Whenever $G_j=1$, the frozen resource is the maximum of the integer
round-ups, so Lemma~\ref{lem:clip} gives $F_j=\widehat p$ there.
Zero-guard each $F_j$ by $G_j$, with bound one, then take their finite
Boolean disjunction, also with bound one. At least one $G_j$ is active;
all active branches give the same correct Boolean value. Thus the result
is $\widehat p\in\B_{1/a}$. No maximum \emph{function} was introduced
into the sparse real algebra. Theorem~\ref{thm:dictionary} gives
$p\in\mathcal P_a$.
\end{proof}

\subsection{Equivalence theorem for the principal problem}
\begin{theorem}[The first relational gap is the exact obstruction]
\label{thm:equivalence}
The following assertions are equivalent; every one quantifies over all
positive finite arities.
\begin{enumerate}[label=(\roman*),leftmargin=2em]
\item All rational bounded-intersection identities \eqref{eq:target} hold.
\item $\B_{1/2}\cap\B_{1/3}=\B_1$.
\item The identity in (ii) holds for Boolean functions.
\item $\mathcal P_2=\mathcal P_1$, or equivalently
$S_2p\in\E_{0,*}\Rightarrow p\in\E_{0,*}$.
\item $\E_{0,*}=\E_{1,*}$ for the explicitly defined integer levels.
\end{enumerate}
\end{theorem}
\begin{proof}
Assertions (i)$\Rightarrow$(ii)$\Rightarrow$(iii) are immediate.
For (iii)$\Rightarrow$(iv), take $p\in\mathcal P_2$. Numerical
monotonicity in Lemma~\ref{lem:monotone} gives $p\in\mathcal P_3$.
The dictionary then puts $\widehat p$ in both bounded source algebras,
and (iii) puts it in $\B_1$. Integer restriction gives $p\in\E_{0,*}$.
The reverse inclusion $\mathcal P_1\subseteq\mathcal P_2$ is already
known from contraction.

Suppose (iv), and take $p\in\mathcal P_a$. Choose $k$ with $2^k\geq a$.
By numerical monotonicity $S_{2^k}p\in\E_{0,*}$. Coordinatewise ceiling
identities give $S_{2^k}p=S_2^k p$. Apply factor-two unpadding $k$ times
to conclude $p\in\E_{0,*}$. Hence every $\mathcal P_a$ equals
$\E_{0,*}$, and Theorem~\ref{thm:union} proves (v).

Finally assume (v). For a coprime pair $a,b$, a Boolean function in
$\B_{1/a}\cap\B_{1/b}$ belongs to $\U_1$ by the unbounded meet.
A Boolean affine grid function is constant on each unit product cell,
so it has the form $\widehat p$. The dictionary gives
$p\in\mathcal P_a\subseteq\E_{1,*}=\E_{0,*}$; hence
$\widehat p\in\B_1$. Theorem~\ref{thm:booleanreduce} yields the full
function identity for $a,b$. For arbitrary $q,v$, put $d=\gcd(q,v)$,
$a=q/d$, $b=v/d$ and use the common conjugacy $\mathfrak T_d$ from
Lemma~\ref{lem:scale}. It transfers the already proved coprime identity
back to \eqref{eq:target}, proving (i).
\end{proof}

\begin{corollary}[Equivalence with the stronger coarsening principle]
\label{cor:coarsening}
The assertions of Theorem~\ref{thm:equivalence} are also equivalent to
\[
 \B_{1/q}\cap\U_{1/d}=\B_{1/d}\qquad\text{for every }d\mid q.
\]
\end{corollary}
\begin{proof}
Assume $\E_{0,*}=\E_{1,*}$. After conjugating by $\mathfrak T_d$, it
suffices to take $f\in\B_{1/a}\cap\U_1$. For $a=1$ there is nothing
to prove. For $a>1$, the endpoint threshold predicate $P$ and tag predicates
$T_i^a$ from the proof of Theorem~\ref{thm:booleanreduce} belong to
$\B_{1/a}$ and are constant on each unit product cell. By the dictionary
and Theorem~\ref{thm:union}, their integer predicates belong to
$\mathcal P_a\subseteq\E_{1,*}=\E_{0,*}$. They therefore belong to
$\B_1$, and the same bounded reconstruction gives $f\in\B_1$.
Conversely the stronger principle, applied with $q=2,d=1$, places every
member of $\B_{1/2}\cap\B_{1/3}$ in $\B_1$, since its unbounded grid
semantics lies in $\U_1$. Theorem~\ref{thm:equivalence} then applies.
\end{proof}

\begin{corollary}[What a negative answer would require]\label{cor:negative}
If $\E_{0,*}\neq\E_{1,*}$, then there is a Boolean counterexample already
for strides $1/2$ and $1/3$. Conversely, any bounded rational-intersection
counterexample proves this predicate separation.
\end{corollary}
\begin{proof}
For an explicit existence argument in the first direction, choose
$p\in\E_{1,*}\setminus\E_{0,*}$. By Theorem~\ref{thm:union} and
numerical monotonicity, there is a least $k\geq1$ with
$S_{2^k}p\in\E_{0,*}$. Put $r=S_{2^{k-1}}p$. Then $r\notin\E_{0,*}$
but $r\in\mathcal P_2\subseteq\mathcal P_3$. Thus $\widehat r$ belongs
to both fine bounded algebras and not to $\B_1$. The other direction
is the contrapositive of Theorem~\ref{thm:equivalence}.
This does not give an unconditional concrete counterexample: selecting
the initial separating predicate is exactly the unresolved ingredient.
\end{proof}

\begin{proposition}[Equivalent doubling-argument closure]\label{prop:doubling}
Factor-two unpadding in Theorem~\ref{thm:equivalence} is equivalent to
closure of $\E_{0,*}$ under $p(\mathbf n)\mapsto p(2\mathbf n)$.
\end{proposition}
\begin{proof}
If doubling-argument closure holds and $S_2p\in\E_{0,*}$, then
$p(\mathbf n)=(S_2p)(2\mathbf n)\in\E_{0,*}$.
Conversely take $p\in\E_{0,*}$ and put $r(\mathbf n)=p(2\mathbf n)$.
Then $S_2r(\mathbf n)=p(2\lceil\mathbf n/2\rceil)$.
The integer function $e(n)=2\lceil n/2\rceil$ is either $n$ or $n+1$;
its parity-controlled selector is in $\E_0$ with earlier bound $n+1$.
Thus $S_2r\in\E_{0,*}$ by composition, and unpadding yields $r\in\E_{0,*}$.
This is closure of a \emph{predicate class}; the doubling function itself
is excluded from $\E_0$ by its one-coordinate growth bound.
\end{proof}

\subsection{An unconditional fragment that is settled}
\begin{proposition}[Eventually periodic unary descriptors]\label{prop:periodic}
Suppose a unary $f\in\B_{1/q}\cap\B_{1/v}$ has an eventually periodic
descriptor sequence on the $1/q$ grid, with the tags and signed offsets
both periodic after some index. Then $f\in\B_{1/\gcd(q,v)}$.
\end{proposition}
\begin{proof}
Let $d=\gcd(q,v)$ and $\delta=1/d$. The established grid gluing obtains
the coarse descriptor at positive coarse index $K$ from fine index
$(K-1)q/d+1$, dividing the offset by $q/d$. If the fine sequence has
eventual period $t$, this subsequence has eventual period
$t/\gcd(t,q/d)$. It uses a finite catalogue of valid coarse descriptors,
including finitely many exceptional initial cells and the zero cell.

The corresponding initial-cell and eventual-residue predicates belong
to $\B_\delta$ by fixed thresholds and a finite-valued residue clock
$D_\delta(y)\bmod T$. Realize each catalogue entry globally by a fixed
constant, a fixed positive translation, or a truncated predecessor.
For some integer $C$, all these unary functions are globally bounded by
$y+C$; choose $C$ above every positive offset and constant in the finite
catalogue. Earlier bounded selectors under this common bound choose
the indicated descriptor. The result is $f$ on every cell.
\end{proof}

\section{The arithmetical complexity of bounded membership}\label{sec:indexset}
We now ask whether an unbounded program's function has some bounded
representation. This extensional question differs from checking a
majorant annotation on the supplied program.

Fix $q\ge1$ and write $\delta=1/q$. Program indices encode finite,
well-formed $\U_\delta$ expressions with fixed positive finite arity.
Every such expression is total. A \emph{candidate bounded description}
additionally contains, at each recursion node, the code of a previously
constructed majorant. Candidates are finite syntactic objects; they need
not satisfy their declared bounds. Let
\[
 \mathrm{BMem}_\delta
 :=\{e:\text{$e$ encodes a $\U_\delta$ program whose function lies in $\B_\delta$}\}.
\]
The ambient coding is any standard effective coding of the finite
expression trees; alternative effective codings give computably
equivalent index sets.

\begin{lemma}[Effective rational witnesses]\label{lem:index-witness}
For fixed rational $\delta=1/q$, finite $\U_\delta$ programs can
be evaluated exactly on nonnegative rational inputs. If two such
programs of equal arity denote different functions, their values
differ at a rational tuple. Likewise, if a pointwise inequality between
the functions of two programs fails at a real tuple, it fails at a
nonnegative rational tuple.
\end{lemma}
\begin{proof}
On every product of the cells $\{0\},((k-1)\delta,k\delta]$,
all intermediate outputs have a single affine descriptor of the form
$c$ or $x_i+c$, with $c\in\delta\mathbb Z$; this is the
cellwise form of the induction proved in Lemma~\ref{lem:coding}.
Two such affine expressions, or their difference, extend continuously
to all included right faces of that product cell. If a strict failure
occurs on such a face, it persists in a relative neighborhood and
therefore at nearby rational points within the same cell. The same
density argument applies directly in the relative interior. Thus
every failed equality or pointwise bound has a rational witness. On singleton zero-coordinate cells use
rational zero coordinates. Direct evaluation at a rational tuple
terminates because each step decreases its recursion address by the
fixed positive rational $\delta$; all intermediate arithmetic is
rational.
\end{proof}

\begin{theorem}[Exact complexity of unbounded program equality]\label{thm:equiv}
For every positive integer $q$, extensional equality of two same-arity $\U_{1/q}$ programs is $\Pi^0_1$-complete under computable many-one reductions. Inequality is $\Sigma^0_1$-complete. The lower bound already holds when testing whether a unary program is identically zero.
\end{theorem}
\begin{proof}
By Lemma~\ref{lem:index-witness}, two programs differ on the real domain if and only if they differ on an effectively enumerable rational input tuple. Exact rational evaluation makes inequality recursively enumerable; hence equality is $\Pi^0_1$.

For hardness, fix an effective enumeration of Turing machines and let $H_e(n)$ be the predicate that machine $e$ halts within $n$ steps on blank input. The family $(e,n)\mapsto H_e(n)$ is uniformly primitive recursive (bounded finite simulation). By the integer-simulation construction of Lemma~\ref{lem:pr-syntax}, there is an effective unary $\U_{1/q}$ program for $G_e(y)=H_e(D_{1/q}(y))$. It is identically zero if and only if machine $e$ never halts. Nonhalting is $\Pi^0_1$-complete, yielding the required reduction. Its complement gives $\Sigma^0_1$-completeness of inequality.
\end{proof}

\begin{theorem}[Exact bounded-membership index complexity]\label{thm:bmem}
For each fixed integer $q\ge1$, the index set
$\mathrm{BMem}_{1/q}$ is $\Sigma^0_2$-complete under computable
many-one reductions. Hardness already holds for binary unbounded
programs. In particular, extensional bounded-algebra membership
is neither decidable nor recursively enumerable nor co-recursively
enumerable.
\end{theorem}
\begin{proof}
\emph{Upper bound.}
Enumerate all finite candidate bounded descriptions $p$ of the
appropriate arity, with stage-ordered annotations for their proposed
earlier majorants. Regardless of whether these annotations are
valid, forgetting them produces a terminating unbounded program.
Choose an effective enumeration of rational tuples for each arity.
There is a decidable predicate $R(e,p,r)$ asserting (i) that $e$
and $p$ have the required syntactic form and arity, (ii) that the
outputs of $e$ and $p$ agree on the first $r$ rational tuples, and
(iii) that every annotated recursion node of $p$ satisfies its
claimed majorant inequality on the first $r$ rational tuples of
\emph{that node's full input arity}. The input enumeration and
majorant checks are effective by Lemma~\ref{lem:index-witness}.
The same lemma implies that universal success on the rational
tuples is equivalent to success on all real tuples. Therefore
\[
 e\in\mathrm{BMem}_\delta
 \quad\Longleftrightarrow\quad
 \exists p\ \forall r\ R(e,p,r),
\]
which is a $\Sigma^0_2$ description.

\emph{Lower bound.}
Let $W_e$ denote the $e$th computably enumerable set in a standard
effective enumeration. The classical index set
$\mathrm{FIN}=\{e:W_e\text{ is finite}\}$ is $\Sigma^0_2$-complete
\cite{SoareFIN}. Choose a standard primitive-recursive stage
approximation $W_{e,n}$, with $W_{e,0}=\varnothing$, and define
$g_e(n)=|W_{e,n}|$, counting distinct elements; the function
$(e,n)\mapsto g_e(n)$ is primitive recursive and monotone in $n$.
By Lemma~\ref{lem:pr-syntax}, uniformly in $e$ one constructs a
binary $\U_\delta$ program computing
\[
 F_e(x,y)=\operatorname{AddInt}\bigl(x,g_e(D_\delta(y))\bigr)
 =x+g_e(D_\delta(y)).
\]
If $W_e$ is finite, $g_e(n)$ is eventually constant, say with
maximum $C$, and so $F_e$ selects among the finitely many
previously available expressions $x,x+1,\ldots,x+C$. The
finite thresholds used for selection are boundedly definable:
for fixed $N\ge1$,
$[D_\delta(y)\ge N]=[\rho_\delta^{[N-1]}(y)>0]$.
Finite nonzero tests and bounded selectors can be implemented
by step recursion with the common previously available majorant
$x+C$. Hence $F_e\in\B_\delta$. This is an extensional
existence argument; there is no requirement to compute $C$
uniformly from $e$.

If $W_e$ is infinite, then $g_e$ is unbounded. By the
one-coordinate domination invariant for bounded programs
(Lemma~\ref{lem:onecoord}), every binary $f\in\B_\delta$
is bounded globally either by $x+C$ or by $y+C$, for some fixed
$C$. The function $F_e$ has neither kind of bound:
fix $x=0$ and let $y$ grow to defeat $x+C$, or fix $y=0$ and
let $x$ grow to defeat $y+C$. Therefore $F_e\notin\B_\delta$.
Thus $e\in\mathrm{FIN}$ if and only if the effectively computed
program code for $F_e$ lies in $\mathrm{BMem}_\delta$.
This gives a computable many-one reduction from $\mathrm{FIN}$.
\end{proof}

\begin{corollary}[All positive rational strides]\label{cor:indexall}
For every reduced rational $s=a/q>0$, extensional equality of $\U_s$ programs is $\Pi^0_1$-complete and the index set of $\U_s$ programs computing functions in $\B_s$ is $\Sigma^0_2$-complete.
\end{corollary}
\begin{proof}
The rational-witness upper bounds apply directly at stride $s$: the local cell boundaries and offsets lie in $(1/q)\mathbb Z$, rational evaluation is exact, and candidate bounded programs admit the same universal rational tests. For hardness, the constructive rational normalization of Theorem~\ref{thm:classification} gives an effective compilation of programs from $\U_{1/q}$ into $\U_s$, while $\B_{1/q}=\B_s$ as extensional classes. The reductions in Theorems~\ref{thm:equiv} and~\ref{thm:bmem} therefore transfer. This argument invokes the rational classification for the extension beyond reciprocal-integer strides; the two original complexity theorems for $1/q$ do not require it.
\end{proof}

\begin{remark}[Independence and scope]\label{rem:index-scope}
The upper bound uses an existential \emph{bounded program} and a
universal test of its semantic majorant obligations; it does not
purport to decide whether an annotated program is admissible.
The reduction uses integer simulation, finite thresholds, and the
one-coordinate bound invariant, but \emph{not} either of the
bounded rational-stride operator compilers or the unproved literal
bounded-intersection identity. Hence the index-set theorem can be
audited independently of the full rational classification.
\end{remark}

\section{Remaining questions and scope}\label{sec:limitations}
The unresolved rational question is whether
$\B_{1/q}\cap\B_{1/v}=\B_{1/\gcdop(q,v)}$.
Theorem~\ref{thm:equivalence} and Proposition~\ref{prop:doubling} identify
it with equality of the integer predicate levels, equivalently factor-two
unpadding or doubling-argument closure. Neither the numerical separation
of $\E_0$ and $\E_1$ nor a bound on the final Boolean output settles it.
The clipped simulation changes precision and does not provide a unit-grid
realization. Clote discusses the low relational gaps as open at the time
of his exposition \cite{Clote1999}.

The classification concerns rational truncated translations with finite
termination. Irrational and nonlinear descents need different invariants;
Zeno descent needs a separate solution semantics. The index-set results
refer to effective syntax and exact rational evaluation, not a procedure
for computing minimal majorants.

\section*{Declaration of generative AI and AI-assisted technologies in the manuscript preparation process}
During the preparation of this work, the author used ChatGPT (OpenAI) for
copyediting, grammatical correction, and language polishing of the manuscript.
The author reviewed and edited the resulting text as needed and takes full
responsibility for the content of the article.

\end{document}